\documentclass[reqno, 10pt, a4paper]{amsart}
\usepackage[dvipsnames]{xcolor}
\usepackage[utf8]{inputenc}
\usepackage[T1]{fontenc}
\usepackage{amsthm}
\usepackage{amsmath}
\usepackage{multicol}
\usepackage{multirow,bigdelim}
\usepackage[utf8]{inputenc}
\usepackage{lmodern} 
\usepackage{tikz} 
\usepackage{caption}
\newcommand{\includegraphicsds}[2]{\raisebox{-.33\height}{\includegraphics[width=#1\textwidth]{#2}}}
\usepackage[labelformat=simple]{subcaption}
\tikzset{
main node/.style={inner sep=0,outer sep=0},
label node/.style={inner sep=0,outer ysep=.2em,outer xsep=.4em,font=\scriptsize,overlay},
strike out/.style={shorten <=-.2em,shorten >=-.5em,overlay}
}

 \usepackage{tikz-cd}
\makeatletter
\ifcase \@ptsize \relax
  \newcommand{\miniscule}{\@setfontsize\miniscule{4}{5}}%
\or
\or
  \newcommand{\miniscule}{\@setfontsize\miniscule{5}{6}}%
\fi
\makeatother

\makeatletter
\ifcase \@ptsize \relax
  \newcommand{\nano}{\@setfontsize\miniscule{3.5}{4.5}}%
\or
  \newcommand{\nano}{\@setfontsize\miniscule{4.5}{5.5}}%
\or
  \newcommand{\nano}{\@setfontsize\miniscule{4.5}{5.5}}%
\fi
\makeatother

    \newcommand{\totimes}{\boxtimes}
\newcommand{\balita}{\raisebox{1.8pt}{\text{ \nano$\bullet$\hspace{1.7pt} }}}

\usepackage{mathtools}

\newcommand{\includegraphicsd}[2]{\raisebox{-.475\height}{\includegraphics[width=#1\textwidth]{#2}}}
\usepackage{textcomp}

\numberwithin{equation}{section}
\newtheoremstyle{mytheoremstyle} 
    {10pt}                    
    {8pt}                    
    {\itshape}                   
    {}                           
    {\scshape}                   
    {.}                          
    {.5em}                       
    {}  

\theoremstyle{mytheoremstyle}

\newtheorem{theorem}{Theorem}[section]
 
 \newtheorem{lemma}[theorem]{Lemma}

 \newtheoremstyle{definition} 
    {8pt}                    
    {5pt}                    
    {}                   
    {}                           
    {\scshape}                   
    {.}                          
    {.5em}                       
    {}  

 \theoremstyle{definition}
 \newtheorem{definition}[theorem]{Definition}
 \newtheorem{example}[theorem]{Example}
 \newtheorem*{acknowledgements}{Acknowledgements}
 \newtheorem{remark}[theorem]{Remark}
 
\newcommand{\includegraphicswextra}[3]{\raisebox{-.#3\height}{\includegraphics[width=#1\textwidth]{#2}}}
 \newenvironment{proof*}[1]{\renewcommand{\proofname}{#1}%
 \begin{proof}}
 {\end{proof}} 
\usepackage{graphicx}
 \definecolor{VerdeFH}{HTML}{009374}

  {\csname align*\endcsname}
  {\csname endalign*\endcsname}

\newcommand{\eff}{^{\mathrm{eff}}_N}  
\newcommand{\uv}{^{\textsc{\scriptsize uv}}_{N}} 

\newcommand{\Nmax}{N_{ \infty }} 
\newcommand{\uvmax}{^{\textsc{\tiny uv}}_{N_{\infty}}} 
\newcommand{\Lebesgue}{_{\textsc{\tiny Leb}}} 
\newcommand{\inter}{^{\textsc{\tiny Int}}}
 
\DeclareFontEncoding{LS1}{}{}
\DeclareFontSubstitution{LS1}{stix}{m}{n}

\makeatletter
  \newcommand*\textmathversion{\csname textmv@\math@version\endcsname}
  \newcommand*\textmv@normal{m}
  \newcommand*\textmv@bold{b}
\makeatother

    \usepackage[margin=0.78in]{geometry}
    \usepackage[bookmarks=false]{hyperref}
\hypersetup{%
    colorlinks=true, linktocpage=true, pdfstartpage=3, pdfstartview=FitV,%
    breaklinks=true, pageanchor=true, 
    plainpages=false, bookmarksnumbered, bookmarksopen=true, bookmarksopenlevel=1,%
    hypertexnames=true, pdfhighlight=/O,
     urlcolor=RoyalBlue, linkcolor=RoyalBlue, citecolor=RoyalBlue,    
}

    \numberwithin{equation}{section}
 
\usepackage{bbm}
\usepackage{footmisc}

    \renewcommand{\and}{\mbox{and}}

    \newcommand{\STr}{\mathrm{STr}}
    
    \newcommand{\mtr}[1]{\mathrm{#1}}

    \newcommand{\A}{\mathcal{A}}
    \newcommand{\dif}[1]{\mathrm{d}#1}

    \newcommand{\re}{\mathbb{R}}

    \newcommand{\diag}{\mtr{diag}}

    \renewcommand{\H}{\mathcal{H}}
    
    \newcommand{\C}{\mathbb{C}}

    \usepackage{amssymb}
    \usepackage{mathrsfs}

    \newcommand{\ee}{\mathrm{e}}
    
    \newcommand{\inv}{^{-1}}

    \newcommand{\N}{\mathbb{N}}

    \newcommand{\hp}[1]{^{(#1)}}

    \newcommand{\where}{\mbox{where}\,\,}

     \DeclareMathOperator{\Day}{\mathscr{D}}
      \DeclareMathOperator{\Hess}{Hess}
    \DeclareMathOperator{\Tr}{Tr}
    \DeclareMathOperator{\im}{im}
    \newcommand{\TrN}{\Tr_{N}}
    \newcommand{\TrA}{\Tr_{\A_n}}

\newcommand{\Cfree}[1]{\C_{\langleb #1\rangleb}}
\newcommand{\Cn}{\Cfree{n}}
\newcommand{\CnN}{\C_{\langleb n\rangleb, N}}

\newcommand{\MN}{M_N(\C)}
\newcommand{\MA}{M_n(\A_n)}

\newcommand{\itemb}{\item[$\balita$]}

\newcommand{\eeqref}[1]{Eq. \eqref{#1}}
\newcommand{\texteqref}[1]{\text{Eq.} \eqref{#1}}

\usepackage{booktabs}
\usepackage{colortbl}
\usepackage{amsmath}
\usepackage{xcolor}
\usepackage{graphicx}
\colorlet{tableheadcolor}{gray!19} 
\colorlet{tablerowcolor}{gray!10} 
   \let\langleb=\langle
   \let\rangleb=\rangle
\usepackage{mathabx}

\usepackage{bbm}

\begin{document}

\title[The algebra of functional renormalization: a ribbon graph derivation]{A ribbon graph derivation of\\ the algebra of functional 
renormalization for \\ random multi-matrices with multi-trace interactions}
 
 \author{Carlos I. P\'erez-S\'anchez}  
 
  \address{Institute for Theoretical Physics, University of Heidelberg \newline \indent 
  Philosophenweg 19, 69120 Heidelberg, Germany, European Union   \newline \indent 
 \hspace{.0cm}\& \newline \indent 
  Faculty of Physics, University of Warsaw  \newline \indent  
  ul. Pasteura 5, 02-093, Warsaw, Poland, European Union   
  }
\email{perez@thphys.uni-heidelberg.de}


\begin{abstract}
We focus on functional renormalization for ensembles of several (say $n\geq 1$) random matrices, whose potentials include multi-traces, to wit, the probability measure contains factors of the form $ \exp[-\mathrm{Tr}(V_1)\times\ldots\times \mathrm{Tr}(V_k)]$ for certain noncommutative polynomials $V_1,\ldots,V_k\in \mathbb{C}_{\langleb n \rangleb}$ in the $n$ matrices. This article shows how the ``algebra of functional renormalization''---that is, the structure that makes the renormalization flow equation computable---is derived from ribbon graphs, only by requiring the one-loop structure that such equation (due to Wetterich) is expected to have. Whenever it is possible to compute the renormalization flow in terms of $\mathrm U(N)$-invariants, the structure gained is the matrix algebra $M_n( \mathcal{A}_{n,N}, \star ) $ with entries in $\mathcal{A}_{n,N}=(\mathbb{C}_{\langleb n \rangleb} \otimes \mathbb{C}_{\langleb n \rangleb} )\oplus( \mathbb{C}_{\langleb n \rangleb} \boxtimes \mathbb{C}_{\langleb n \rangleb})$, being $\mathbb{C}_{\langleb n \rangleb} $ the free algebra generated by the $n$ Hermitian matrices of size $N$ (the flowing random variables) with multiplication of homogeneous elements in $\mathcal{A}_{n,N}$ given, for each $P,Q,U,W\in\mathbb{C}_{\langleb n \rangleb}$, by \begin{align*}(U \otimes W) \star ( P\otimes Q) &= PU \otimes WQ \,, & (U\boxtimes W) \star ( P\otimes Q) &=U \boxtimes PWQ \,, \\(U \otimes W) \star ( P\boxtimes Q) &= WPU \boxtimes Q \,,\ & (U\boxtimes W) \star ( P\boxtimes Q) &= \mathrm{Tr} (WP) U\boxtimes Q \,,\end{align*} which, together with the condition $(\lambda U) \boxtimes W = U\boxtimes (\lambda W) $ for each complex $\lambda$, fully define the symbol $\boxtimes$.\end{abstract} 
 \fontsize{10.55}{13.45}\selectfont    
\maketitle

\section{Introduction and motivation} \label{sec:intro}
By the \textit{Functional Renormalization Group} (FRG) 
physicists refer to a certain flow in the renormalization time $t$, usually the logarithm $t=\log k$ of the energy scale $k$, which in the ``nonperturbative'' \cite{Berges:2000ew} setting is governed by 
Wetterich equation \cite{Wetterich} 
\begin{align*}``\partial_t \Gamma _{k}[\phi]= \frac12 \STr \bigg(\frac{\partial_t R_{k} }{ \Hess \Gamma_{k}[\phi]+R _{k}}\bigg)\text{''}\,. \end{align*}
This is satisfied by the ``effective action $\Gamma_k[\phi]$,
infrared-regulated by $R_k$ up to the energy scale $k$, on some space of fields $\phi$'' 
(quotation marks, since mathematical details follow for the system of our interest).
This article addresses functional renormalization for ensembles of $n$-tuples of 
Hermitian matrices; the particular type of ensembles we analyze 
have clear physical motivations (Secs. \ref{sec:ncpol} and \ref{sec:multitr}).  

While there is no better way to compute it, the denominator in the right hand side of Wetterich equation is a Neumann expansion (geometric
series) in ``$\Hess \Gamma_k [\phi] /R_k$'', essentially, the  Hessian  of the fields. 
For matrix ensembles, this Hessian is an object of four indices, two
from each of the two derivatives. The question 
is which is the meaning of the product $\star$ implied in powers $(\Hess \Gamma_k[\phi])^{\star m} $  of the Hessian; we call the  
algebra defined by such product the \textit{algebra of functional renormalization}\footnote{We should probably write the \textit{functional renormalization group algebra},
but ``group algebra'' can be confused with $\C[ G]$, for a group $G$; 
or \textit{functional renormalization algebra}, which would suggest,
that the renormalization group (which is none) is upgraded to an algebra.}. \par  Of course, this question can be answered directly by looking at 
the proof of Wetterich equation; for the field theory in question, see  \cite{FRGEmultimatrix}. For multi-matrix
ensembles with probability measures
defined, as is usual, solely in terms of single traces of matrix polynomials, 
part of the answer relies on observing that 
the Hessian is spanned,  as  $\sum_\alpha F_\alpha \otimes G_\alpha $, by
couples of noncommutative polynomials $F_\alpha,G_\alpha$. The (so far, unsurprising) answer is that powers of the Hessian are obtained by the product rule 
\begin{align*}
( U\otimes W)\star  (P \otimes Q) = PU  \otimes WQ   \,.
\end{align*} 


\begin{figure}[htb]
 \begin{align*} {\underbrace{\includegraphicsd{.230}{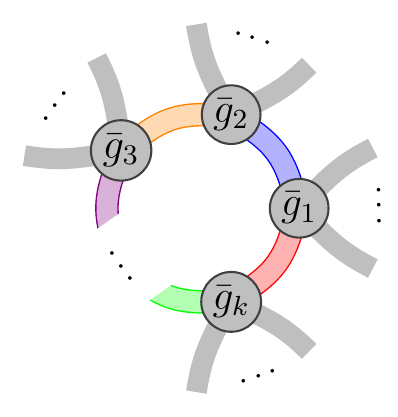}}_{\substack{\text{A one-loop diagram in a simple case}\\ \text{where ``all legs are pointing outwards''} }}} \hspace{1.5995ex}  {  \stackrel{ \substack{\text{\tiny } \\ \text{ }}}{\mapsto} \hspace{2.995ex} \underbrace{
\includegraphicsd{.230}{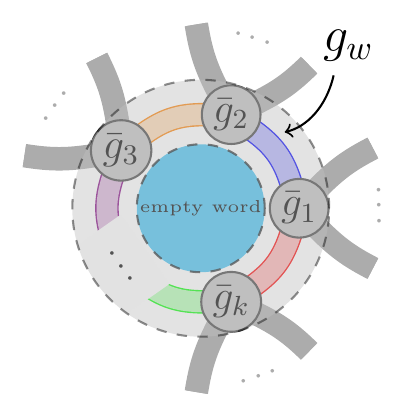}}_{\text{$ {\color{cyan!80!black}1_N\,}\otimes $ cycl. outer word $w$ $(w\in \Cfree{n}$)}}} 
\end{align*}\centering 
\caption{\label{fig:ideal} The colored legs correspond with $ M_n$-block 
entries $\Hess_{\color{PineGreen}a\color{black},\color{red!80!black}b} \Gamma$
of the effective action $\Gamma$.   \textit{Left:} Unrenormalized interactions $\bar g_i$ appearing in a $k$-th
power of the Hessian. \textit{Right:} The contribution to the $\beta_w$-function, $w$ formed by reading off clockwise the legs}
\end{figure}

Notice the ``inversion'' in one of the first tensor-factors, which starts to 
reflect the inner boundary and the outer one of the one-loop, 
relevant in this note (see Figures \ref{fig:ideal} and \ref{fig:notideal}). Interestingly,  the incorporation 
of double traces yields a less trivial answer, 
for a ``second product'' appears (if one wants, 
a twisted tensor product) that  
also satisfies bilinearity $( z P) \totimes Q = P \totimes (zQ )$, $z\in\C$, but which differs from the usual tensor product 
only in the way one multiplies it with another element, $U\otimes W$ or $T\totimes V$.
From interactions of two (or more) traces, then 
the Hessian of the effective action turns out to be spanned by noncommutative 
polynomials in a more general position\footnote{These, moreover, might 
have traces as factors, but being these scalar functions 
of the matrices that will not be derived again, this is irrelevant} $\sum_{\alpha} F_\alpha \otimes G_\alpha 
+ \sum_\rho H_\rho \totimes I _\rho$.
The product reads 
\begin{subequations}\label{firstFRG}%
\begin{align}%
(U \otimes W+ Y\totimes Z) \star ( P\otimes Q)  &=  PU \otimes WQ + Y \totimes PZQ \,, \\
(U \otimes W+Y\totimes Z) \star ( T\totimes V)  &=  WTU \totimes V+ \TrN (WT) U\totimes V\,. 
\end{align}\end{subequations}%
The aim of this article is to prove,
using graphs, that the sole assumption
that the contributions to the rhs of Wetterich  equation have \textit{all} a ``one-loop
structure'' implies that the rhs of Wetterich equation 
satisfies \eeqref{firstFRG}. Next, we justify  the appearance of 
noncommutative (nc) polynomials and 
of double traces, relating both with other theories (in Secs. \ref{sec:ncpol} and \ref{sec:multitr}, respectively). 
In Section \ref{sec:terms_mainstatement}, before presenting the precise statement,
we give a short, but self-contained account of the ribbon graph theory 
needed to prove, in Section \ref{sec:proof}, the main statement (Thm. \ref{thm:main}).
Appendix \ref{app:RN} explains the construction of the infrared-regulated effective action.

\subsection{The origin of the noncommutative polynomials and potential applications}\label{sec:ncpol}
Ensembles of several matrices with probability 
laws given by ordinary (commutative) potentials
are extensively studied in high energy physics. 
An important family of models solved by Eynard-Orantin \cite{Eynard:2005iq}, using their topological recursion, is the \textit{two-matrix model}, which refers to measures $\dif\mu$ on $\H_N^2$ of the form
\begin{equation}
 \label{factorizable}
\dif\mu(A,B)= \exp[-\TrN(AB)] \exp\big\{-\TrN [V_1 (A)\big\} (\dif A)\Lebesgue 
\times \exp\big\{-\TrN [V_2 (B)]\big\} (\dif B)\Lebesgue\,.  
\end{equation}
Modulo the first factor, this is still a product of measures, each of which on 
 the space $\H_N$ of $N\times N$ Hermitian matrices.
Here, $V_1(x) $ and $V_2(x)$ are polynomials in a real variable $x$ 
and $\TrN(X)=\sum_{i=1}^N X_{i,i}$ is the unnormalized trace. We will keep this 
notation in the sequel. \par 
The simplest addressed and (using the character expansion method \cite{KazakovABAB}) solved model 
with a genuinely noncommutative law, generalizing \eeqref{factorizable}, is the $ABAB$-model
with measure \begin{align} \nonumber
\qquad\dif \mu (A,B) &= \exp\big\{-N \TrN (g_{A^4} A^4+ g_{B^4} B^4 + g_{ABAB} ABAB) \big\} \dif \gamma (A, B) \\&  =:\exp(-S\inter[A,B])\dif \gamma (A, B) , \label{ABABmeasure}
\end{align}
where  
\begin{equation}
 \dif\gamma(A,B)= \exp\Big\{-\frac{N}{2}\TrN (A^2+B^2)\Big\} (\dif A)\Lebesgue (\dif B)\Lebesgue  
\end{equation} is the product Gaussian measure on $\H_N^2$.
The action $S$ that defines the probability measure $\dif\mu=\exp(-S[A,B])(\dif A)\Lebesgue (\dif B)\Lebesgue$ is
the \textit{bare action}.
Hermitian ensembles with wildly non-factorizable measures, as those relevant in this paper, generalizing
\eqref{ABABmeasure}, are studied in
free probability \cite{GuionnetFreeAn}. \par 
A more recent  application of nc polynomial interactions concerns 
ensembles of Dirac operators  \begin{equation}
                        \label{diracens}
\mathcal Z=\int_{\mtr{Dirac}} \exp[-S(D) ] \dif D
                              \end{equation} which aim at the quantization of the spectral action $S(D)=\Tr f(D) $ in 
noncommutative geometry. This problem was posed since \cite[\S 19]{ConnesMarcolli}  
and finite approximations to smooth geometries that allow to make precise 
sense of the partition function \eqref{diracens} recently reawakened interest in the problem  \cite{GlaserSternZwei,MultimatrixYMH,Khalkhali:2021eso}.
Random (finite) noncommutative geometry was first approached with 
 Monte-Carlo simulations for spectral observables
 of ensembles of finite rank Dirac operators $D$ of geometries of arbitrary signature,
 say, with $p$ plus and $q$ minus signs. These finite-dimensional 
spectral triples based on the matrix algebra $M_N(\C)$ are known as fuzzy or matrix geometries. It was algorithmically convenient \cite{BarrettGlaser}---and with  
combinatorial arguments systematically
possible \cite{SAfuzzy}---to parametrize 
the Dirac operators (matrices of size $ kN^2 \times kN^2$, with $k=2^{p+q-1}$) in terms of smaller matrices. After solving the fuzzy spectral triple axioms \cite{BarrettMatrix}, Dirac operators take the form  \begin{align*}D=\sum_{\mu=1}^p \gamma^{\mu}\otimes  [X_\mu, \balita ]+ \sum_{\mu=p+1}^{p+q} \gamma^{\mu}\otimes \{X_\mu, \balita \} + \text{higher $\gamma$-products}\,,   \end{align*}
now parametrized in terms of 
commutators and anti-commutators of $N\times N$  matrices.
Tracing the powers of $D$ yields nc polynomial interactions 
spanned by Hermitian $X_\mu$ ($1\leq \mu \leq p$) and anti-Hermitian
($1\leq \mu-q \leq p$) matrices\footnote{Further matrices 
appear in all higher (odd) orders, but due to the Clifford algebra $ \{\gamma^\mu, \gamma^\nu\}=2 \eta^{\mu\nu} 1_{k\times k}$  the expansion terminates
at some parametrizing matrix $X_n$. Also the trace disperses the $\gamma$-matrices but these are important, since they determine the coefficients in the nc polynomials in $X_1,\ldots,X_n$.} (and those appearing with higher-degree products 
of gammas) as well as double traces, cf. \eeqref{measures} below.
\par

\begin{figure}
 \begin{align*} {\underbrace{\includegraphicsd{.230}{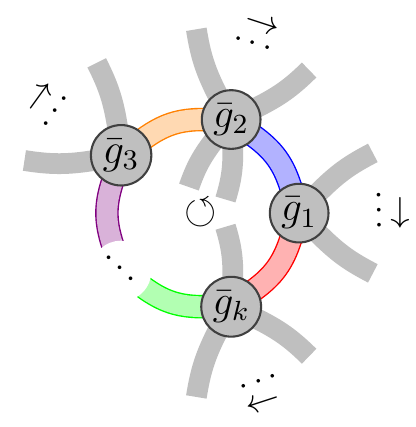}}_{\text{General situation (yet, without multi-traces)}}} \hspace{4.995ex}  {  \stackrel{ \substack{\text{\tiny} \\ \text{ }}}{\mapsto} \hspace{1.995ex} \underbrace{
\raisebox{-.52\height}{ 
\includegraphics[width=.17\textwidth]{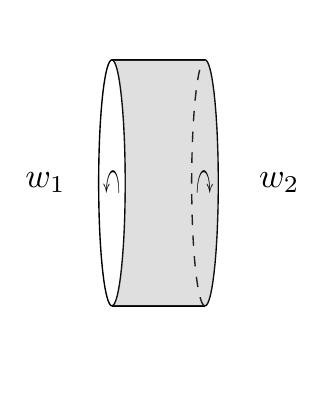}} }_{\text{inner word $w_1$  $\otimes $ outer word $w_2$}}} 
\end{align*}\centering 
\caption{\label{fig:notideal} How the one-loop structure of the FRG is encoded in $M_n(\A_n,\star)$.   
\textit{Left:} Unrenormalized interactions $\bar g_i$ appearing in a $k$-th
power of the Hessian. \textit{Right:} Unlike Figure \ref{fig:ideal}, 
this situation leads to a cylindric topology. Each word $w_1$ and $w_2$ 
distributed at the boundary is oriented in a consistent way with an 
orientation of the cylindric surface (determined by the cyclic clockwise order 
in the interaction vertices).}
\end{figure}
Applications of nc polynomial interactions were relevant for a better understanding of the Temperley-Lieb algebra\footnote{I thank Bertrand Eynard for pointing out the Temperley-Lieb algebra in the context of nc polynomial matrix interactions.}. From a  Temperley-Lieb vertex $\mathcal B$, i.e. a rooted, planar chord diagram 
one obtains a nc polynomial by distributing matrices $X_{l_1},\ldots,X_{l_{2r}}$ at the chord ends, and summing 
over the remaining indices after placing $r$ Kronecker deltas, namely a
$\delta_{l_t}^{l_s}$
for each chord joining the $t$-th node (clockwise from the rooting $*$) 
with the $s$-th node. For instance, 
\[
\mathcal B= \hspace{-.2 cm} \includegraphicsd{.102}{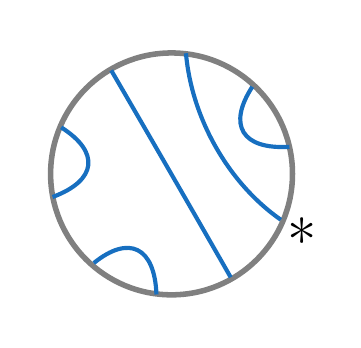}   \hspace{ -.0 cm}  \mapsto    \hspace{ .2 cm} \mathcal B(X_1,\ldots, X_n) = \sum_{a,b,c,d,e=1}^n X_a X_b X_c^2 X_d^2 X_b X_a X_e^2  \,.
\]
Nc polynomial matrix interactions 
are also auxiliary in the description of more general planar algebras  \cite{Guionnet:2010qg} and $O(\mathfrak n)$-loop models. 

\subsection{On multi-trace interactions}\label{sec:multitr}
 
We will see later that not including multi-trace
interactions in renormalization is unnatural (since generic 
radiative corrections include more traces than the  
bare action did). This short section mentions theories that 
contain multi-traces even before addressing renormalization (whenever possible).
 \par

\begin{itemize}\setlength\itemsep{.64em}%

\itemb  \textit{Dirac ensembles}
 always yield double trace interactions (then 
 renormalization creates even more traces). The (bare) Dirac ensemble measure  
 is of the type
\begin{equation} \label{measures}
 \qquad\dif{\mu} (X_1,\ldots, X_n )= \exp\big\{ -N \TrN (P)  -  \TrN^{\otimes 2 }  (Q_{(1)}\otimes Q_{(2)} ) \big\} \dif (X_1,\ldots, X_n )\Lebesgue
\end{equation}
for $P$, $Q_{(1)}$ and  $Q_{(2)}$ also\footnote{This 
notation has been inspired by Sweedler notation in quantum groups,
and avoids to write sums like $Q_{(1)}\otimes Q_{(2)}=\sum_{\alpha} Q_{1,\alpha }\otimes  Q_{2,\alpha}$ where each $Q_{1,\alpha},Q_{2,\alpha} \in \Cn$} nc polynomials in 
the matrices $X_1,\ldots, X_n$ and $
\dif (X_1,\ldots, X_n )\Lebesgue$ is the product Lebesgue measure, now on $\H_N^n$. 
Even though $P$ has an extra factor of $N$ 
with respect to the double trace, observe that the latter cannot be neglected, since $\TrN(Q_{(1)}) \times \TrN (Q_{(2)})$
contains a double sum too. 
 
\itemb \textit{Face-worded, stuffed maps.} Combinatorial 
maps (``gluing of polygons'' dual to ribbon graphs) are counted with the aid of matrix partition functions \cite{PlanarDiags}. In the 
presence of two random matrices with the probability 
law \eqref{factorizable}, the faces of these maps can be uniformly colored 
(and interpreted as Ising model) \cite[\S 8]{EynardCounting}. If the 
potentials are noncommutative polynomials, this is no longer possible,
and the partition function generates maps whose faces are labelled by ``cyclic words'' in the matrices (thus
 the maps could be called \textit{face-worded}, as presented 
 in Figure \ref{fig:wordedmap} for the alphabet $\{A,B\}$).
If the interaction vertices have several traces,
the generated maps are said to be \textit{stuffed} \cite{stuffed,Khalkhali:2021eso} (independent of whether
the potentials are ordinary or noncommutative). 
The terminology reflects that one now allows maps to have elementary cells 
of a topology that need not be that of a disk, i.e.
one has ``maps stuffed with  bordered Riemann surfaces''.
The renormalization flow we study yields equations 
for the $\beta$-functions for matrix ensembles 
whose partition function generate 
``face-worded, stuffed maps''. The fixed-point solution
of the $\beta$-function system \eeqref{betaeta} could be useful to compute 
critical exponents (see Remark \ref{rem:questions}). 

 \itemb \textit{``Touching interactions''}. In several quantum
 gravity approaches, multi-trace operators appear, to name only few:
 \vspace{1ex}
\begin{itemize}\setlength\itemsep{.64em}%
\item in \textit{Liouville gravity}, multi-trace 
one-matrix models are interpreted as generating functions of 
  surfaces that might touch at isolated points. (The planar sector, for instance, 
is grasped, according to  \cite{Klebanov:1994pv}, as trees of spheres
that can touch other spheres at most once.)

\item   multi-trace interactions appear in
\textit{curvature matrix models} \cite{DasMultitrace}. Double traces appear in 
the effective description of a matrix model
with a kinetic term $\Tr(\phi E\phi E)$ (with broken symmetry by a constant matrix $E$).

\item another interpretation in terms of wormholes
appears in (a certain two-matrix model description of) 3-dimensional \textit{Causal Dynamical Triangulations} \cite{Ambjorn:2001br}

\item under the \textit{AdS/CFT}-correspondence,
the AdS-object matching multi-trace operators in CFT are multi-particle
states. In this context, for those states \cite{Witten:2001ua} defines the natural boundary 
conditions at $\infty$.  
\end{itemize}

\end{itemize}

\begin{figure}
 \raisebox{-.33\height}{\includegraphics[width=.86\textwidth]{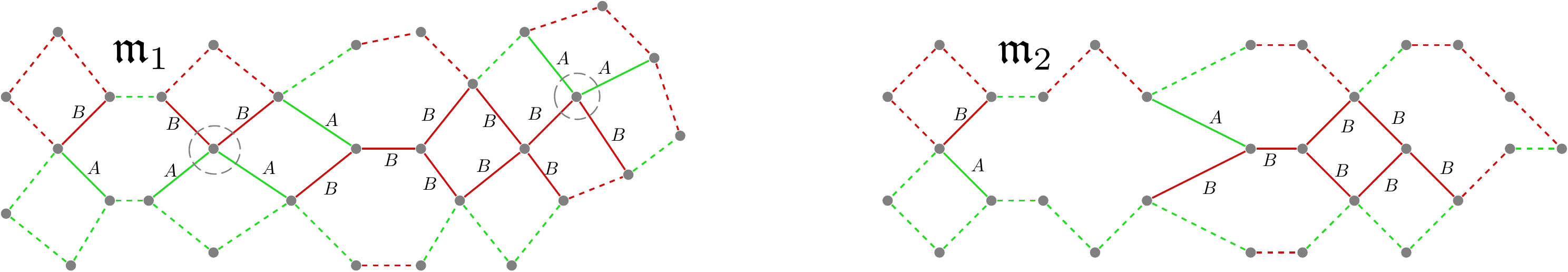} }
 \caption{Example of \textit{(face-)worded maps}  $\mathfrak m_1$ and $\mathfrak m_2$ 
 dual to ribbon graphs generated by multi-matrix models with noncommutative
 polynomial interactions. Each $r$-agon of sides marked with letters $X_{i_1}\ldots X_{i_r}  $ 
is generated by the interaction vertex $\TrN(X_{i_1}\ldots X_{i_r} ) $. 
The relation between  $\mathfrak m_1$ and $\mathfrak m_2$ is the
 renormalization flow. In the cross graining 
 process the one-loop configurations at the nodes 
 marked with dashed circles in $\mathfrak m_1$  yield the effective (in this case, higher-degree) 
 interactions in  $\mathfrak m_2$. This is the dual version of the
 cross graining depicted in Figure \ref{fig:ideal}.
 (Dashed 
edges mean that the maps can extend in that direction and might get some 
non-planar topology). \label{fig:wordedmap}}
\end{figure}

 \tableofcontents

\section{Terminology and main statement}\label{sec:terms_mainstatement}
  
Since our aim is to connect
combinatorics and algebra in matrix models on the one hand,
with renormalization on the other, this article is somewhat 
interdisciplinary. Therefore, it is convenient to 
precisely define our framework and notation. 

\subsection{Ribbon graphs and the noncommutative Hessian on single trace interactions}

The next points introduce our notation and present some definitions:
\begin{itemize}\setlength\itemsep{.64em}%
\itemb The space of Hermitian $N\times N$ matrices is denoted by $\H_N$. The size of matrices 
$X_1\hp N,X_2\hp N,$$\ldots,$ $X_n\hp N\in \H_N$ (which will become the random variables) 
will be relevant, but the lighter notation $X_1,X_2$, $\ldots,$ $X_n $ is convenient. The number $n$ of matrices remains fixed
and we will denote the $n$-tuple $(X_1,X_2,\ldots, X_n )$ by $\mathbb X$.  

 \itemb $\Cfree{n}=\C\langleb X_1,X_2,\ldots, X_n \rangleb$ is the \textit{free algebra}. Any 
 element of $\Cn$ is spanned by words in the alphabet $\mathbb  X$, and 
 $\Cn$ is endowed with the concatenation product. 
 We actually should write $\CnN$ instead of $\Cn$ emphasizing that the 
 generators $X_a$ are matrices of size $N$, but the only manifestation of it is 
 the empty word being the unit matrix $1_N$, and we opt again for a light notation. 

\itemb The \textit{noncommutative derivative} with respect to $A$,
$  \partial_{A} :\Cfree{n} \to \Cfree{n}^{\,\otimes \,2}$ on a word $w$ containing $A$ is the sum over ``replacements of $A$ in $w$ by the $\otimes$ tensor product symbol'' in 
the middle of the word; if $A$ occurs at the left (resp. right) end, then one additionally attaches 
the empty word (or in $\CnN$ a unit $1_N$)
to the left (resp. right) of $\otimes$. For example, in a free algebra with enough letters
\begin{align*} 
\partial_{{A}} ({P\color{Gray}AA\color{black}R})   &=  P \otimes {AR}+ {PA} \otimes   R\,, \\\text{but }\quad 
\partial_{{A}} ({\color{Gray}A\color{black}LGEBR\color{Gray}A}) &  = 1\otimes {LGEBRA}+ {ALGEBR} \otimes 1 \,. 
\end{align*}  
\itemb The noncommutative derivative 
defined on ``cyclic words'' $\Tr P$,  $P\in \Cn$,
is given by  the sum of all possible excisions $P\setminus A$ of $A$ from $P$, rooting (i.e. starting) the remaining word at the letter after the removed $A$
\begin{align*}\partial_A:\im \Tr  \to \Cfree{n} , \quad   \Tr P \mapsto  \sum_{\substack{\text{rootings at}\\ \text{$A$'s next letter}}} P  \setminus A\,.\end{align*}
The result $\partial_A \Tr P=: \mathscr D_A P$ defines the \textit{cyclic derivative} $\mathscr D_A$ of $P$ and is due to Rota-Sagan-Stein \cite{Rota}  and Voiculescu \cite{Voi_gradient}.
For instance, $ 
\partial_{{A}} \Tr({P\color{Gray}AA\color{black}R})   = {ARP}+ {RPA} = \mathscr D_{ A} ({PAAR})$. 
The adjective ``cyclic'' for $\mathscr{D}$ comes from the property $\Day_{X_a} P = \Day_{X _a}[\sigma (P)] $,
which holds for any cyclic permutation  $\sigma(P)$ of the letters of $P$ ($P\in \Cn$ and any $a=1\ldots,n$).

\itemb Grasping $\Tr$ as the trace in $\Cn$
induced by that of $M_N(\C)$, define the \textit{noncommutative Hessian} \cite{FRGEmultimatrix} of a cyclic word 
\begin{align}
\quad\Hess: \im \Tr &\to M_n( \Cn\otimes \Cn)   \nonumber\\
 \Tr P &\mapsto (\Hess_{a,b} \Tr P)_{a,b=1,\ldots,n}:=(\partial_{X_a}\circ\partial_{X_b} \Tr P)_{a,b=1,\ldots,n} \label{NCGHess}\,.
\end{align}
Referring to the block $M_n$-matrix structure, i.e. to indices $a,b=1\ldots,n$, notice that in general the nc Hessian is not a symmetric matrix, $\Hess_{a,b} \Tr P \neq \Hess_{b,a} \Tr P$. 

\end{itemize}

The $(b,a)$-entry in the $M_n$-matrix block structure 
of the Hessian of a cyclic word $\Tr W$ can be represented graphically by summing over all the \textit{ordered}
double markings of $X_a$ and $X_b$ inside a word $W$. 
On $ W= X_{\ell_1}X_{\ell_2} \cdots X_{\ell_{k}}\in \CnN $  
(with $k\geq 2$), according to \eeqref{NCGHess}, this is given 
for $a,b=1,\ldots,n$ by (for a proof see \cite[Prop. 2.3]{FRGEmultimatrix})
\begin{align}\label{doubleNCder}
(\partial_{X_b} \circ \partial_{X_a}) \TrN W=
\sum_{\pi=(uv)} 
\delta^{a}_{\ell_u} 
\delta^{b}_{\ell_v} \pi_1(W) 
\otimes
\pi_2(W)  = \sum_{\pi=(uv)} 
\delta^{a}_{\ell_u} 
\delta^{b}_{\ell_v}  \includegraphicsd{.172}{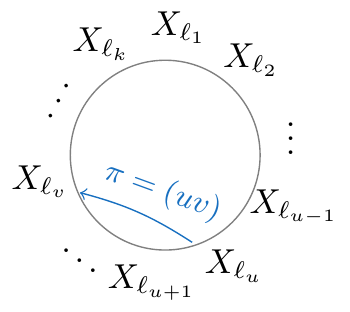} \,,
\end{align}
We sum over all oriented pairings $\pi=(uv)$
between the letters of the cyclic word $\Tr W$  (which explains the circle in the second equality). 
In Eq. \eqref{doubleNCder},  $\pi_1(W) $ is the ordinary 
word between $ X_{\ell_u} $ and $X_{\ell_v} $ and $\pi_2(W)$ that between  $ X_{\ell_v} $ and $X_{\ell_u}$, and because of the deltas $X_{\ell_v} =X_b$ and $X_{\ell_u} =X_a $ must hold,
and the empty word in either case leads to writing $1_N$.

\begin{example}\label{ex:cuts}
To simplify the drawings, we expose the case $n=2$.
We compute the nc Hessian entry corresponding to 
$X_a=A$ and $X_b=B$ on $\TrN W= \TrN (ABAABABB)$.
The entry   reads:
\begin{align*}
\Hess_{b,a} (\Tr W)& =
 \partial _B \partial_A  \bigg( \includegraphicsd{.102}{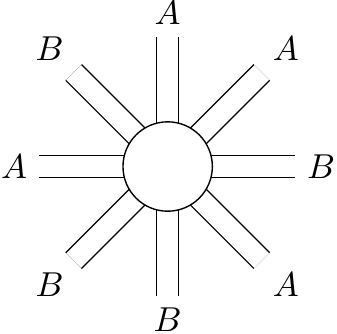}
 \bigg) \\ & = 1_N\otimes \bigg(
 \includegraphicsd{.1}{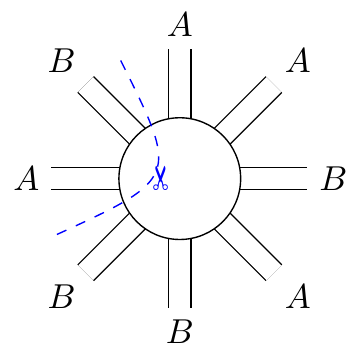}
+\includegraphicsd{.1}{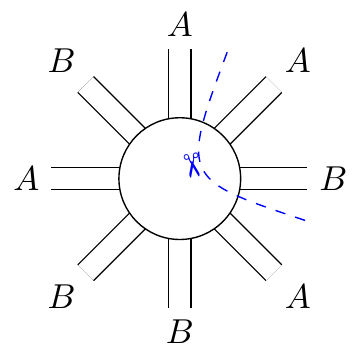}+
\includegraphicsd{.1}{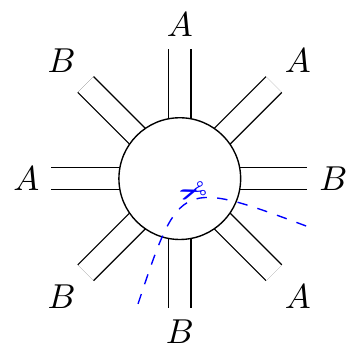} \bigg)\\ 
& + \bigg(
\includegraphicsd{.1}{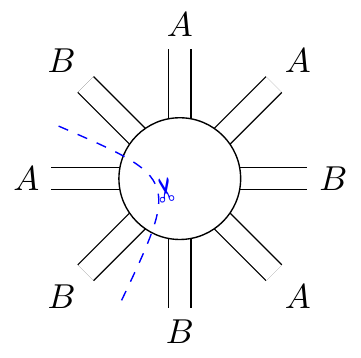}+
\includegraphicsd{.1}{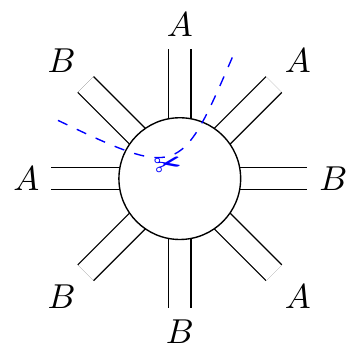}+
\includegraphicsd{.1}{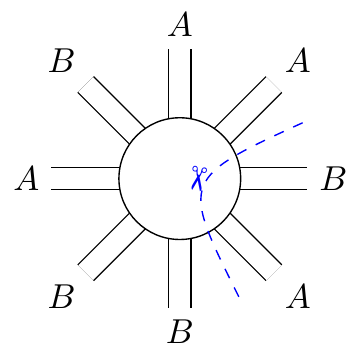}\bigg) \otimes 1_N  
\\
& + \text{polynomials of the form $P\otimes Q$ with $Q\neq 1_N\neq P$}
\end{align*}   
The cyclicity is lost due to each cut (dashed line).
The word represented by each excision is read starting from the letter
right after\footnote{``After'' is determined by the scissors pointing in 
that direction. This is clearer in the graph in \eeqref{tijeras}.} the cut: the first one is $AABABB$, $\ldots $ ,
and the sixth $BBABAA$.  
These terms that are listed arise from contiguous appearances of $AB$ and $BA$ in $W$
and in each case the empty word between the letters originates the $1_N$ tensor factor. 
According to \eeqref{doubleNCder}, the rest of the 
polynomials (last line) are computed by cutting the circle
into two non-trivial words. For instance the next cut yields
\begin{align}\label{tijeras}
\raisebox{-.46\height}{\includegraphics[height=2.3cm]{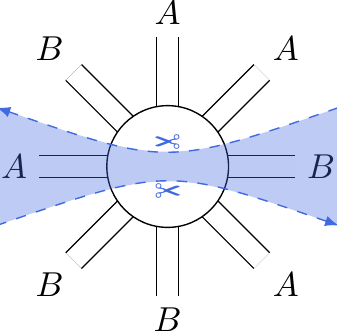}}
\to  BAA \otimes ABB 
\end{align}
The order of the derivatives $\partial _B \partial_A $ (to the left of the cut ``from $A$ to $B$'', $\Rightarrow$ first factor, to the left of ``from $B$ to $A$'' $\Rightarrow$ second factor)
determines which word is placed in which tensor factor.
\end{example}

\subsection{Multi-trace interaction vertices, effective vertices }

The interaction vertices 
in the measures $
 \dif{\mu}(\mathbb X)= \exp\big\{ -N \TrN (P)  -  \TrN^{\otimes 2 }  (Q_{(1)}\otimes Q_{(2)} ) \big\} \dif\gamma(\mathbb X)$ 
are represented by ribbon vertices framed with a dashed circle. This is unusual,
but in view of the multiple products of traces in the measure, a helpful notation. 
The coupling constant  $\bar g$ of multiple trace interactions is what 
prevents the multiple traces from being interpreted as different, 
disconnected polygonal building blocks (and are interpreted as ``touching-interactions'' \cite{DasMultitrace,Klebanov:1994pv,Ambjorn:2001br,Witten:2001ua} in 
other settings). 
Their relation to the free algebra is explained with the following examples (where green/light means
the $A$ matrix and red/dark represents $B$)
\begin{align}\bar g_1\TrN (ABBBAB) \quad\leftrightarrow &\quad\includegraphicsd{.084}{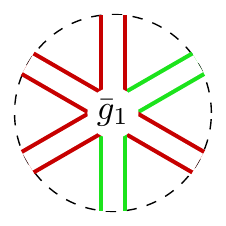}\,\,\\
\bar g_2\TrN^{\otimes 2}(AABABA\otimes AA )\quad
\leftrightarrow& \quad \includegraphicsd{.084}{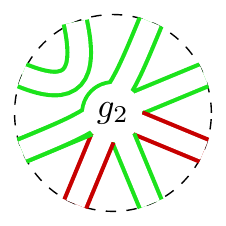}  \\
\bar g_3\TrN^{\otimes 2}(BBABB \otimes A )\quad
\leftrightarrow& \quad \includegraphicsd{.084}{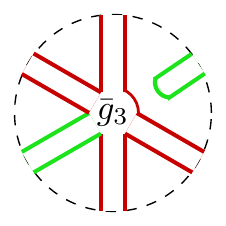} 
\end{align}
The convention is that the label of the coupling constants applies to everything inside the dashed circle, i.e. simultaneously both traces (see also Example \ref{ex:1loop}). 
This representation also reflects the mathematical nature of the effective action  $\Gamma_N[\mathbb X]$  as (for now, at least) a formal series (with the coupling constants as parameters) of the form 
\begin{align} \label{GammaFormalSeries}
\Gamma_N[\mathbb X]= \sum_\alpha O_\alpha \qquad  O_\alpha=\bar g_{\alpha} \prod_{r=1}^{t_\alpha} \Tr_N(w_{\alpha,r})\,,\quad  \text{ monomials }w_{\alpha,r} \in \Cn = \C\langleb \mathbb X \rangleb \end{align}
so $t_\alpha$ is the number of
traces in the \textit{operators} $O_\alpha$. The monomials $ w_{\alpha,r} $ need not be monic; as a matter of fact,
one usually normalizes $w_{\alpha,r}$  with symmetry factors. The coefficient of the kinetic operator $\Tr ( X_c^2/2)$ (for each $c=1,\ldots,n$)  
is called the wave function renormalization (of the matrix $X_c$) 
and, since it is special, it is usually denoted not by a $\bar g$ but by  $Z_c$.  Else, we call  \textit{interaction vertices} the remaining $O_\alpha$'s. 
The bar on the coupling constant $ \bar g_i=\bar g_i(N)$, which are functions of $N$,  denotes that it will still be
rescaled $\bar g_\alpha \to g_\alpha=Z^{\lambda_\alpha} N^{\kappa_\alpha} \bar g_\alpha$, solving 
for $\lambda_\alpha$ and $\kappa_\alpha$, in order to render finite and $Z$-independent  the next
system (only in the large-$N$ differential\footnote{I thank Alexander Schenkel for 
pointing out that the parameter $t$ is still discrete for finite $N$ and 
thus \eeqref{betaeta} is not yet a system of differential equations.}) equations
\begin{subequations}\label{betaeta}
\begin{align}
\big\{\eta_c:\!&=-\partial_t \log Z_c = -Z_c\inv \times \text{coeff. of } \Tr_N(X_c^2/2) \text{ in rhs of \eeqref{Wetterich}} \big\}_{c=1,\ldots,n}
\\ 
\Big\{ \beta_\alpha:\! &= \partial_t g_\alpha  = \text{coefficient of }\prod_{r=1}^{t_\alpha} \Tr_N(w_{\alpha,r}) \text{ in the rhs of \eeqref{Wetterich}} \Big\}_{\alpha}\,,\quad t=\log N\,
\end{align}
\end{subequations}%
of ($\eta$-functions and) $\beta$-function equations
for the interaction vertices $\alpha$, 
determined by Wetterich equation. 
This list of operators appearing in Eqs. \eqref{betaeta} includes those of the original (bare) 
action $S$, but additionally those generated from it by ``radiative corrections'' to $S$.
For instance\footnote{These graphs are based on the next comment: \cite{komentarzABAB}.} 
if the initial model is given by  \begin{align} \label{ABABaction}
S=N \TrN\Big\{ \frac12 A^2+\frac12 B^2 + g_{A^4}\frac{1 }{4}A^4+ g_{B^4}\frac{1}4 B^4 +\frac12 g_{ABAB} ABAB \Big\}
                                  \end{align}  then the radiative corrections 
\begin{align}  
\includegraphicswextra{.0691}{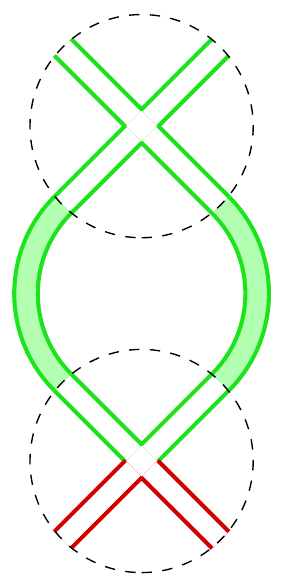}{412} \,,\hspace{.051\textwidth}
\includegraphicswextra{.0691}{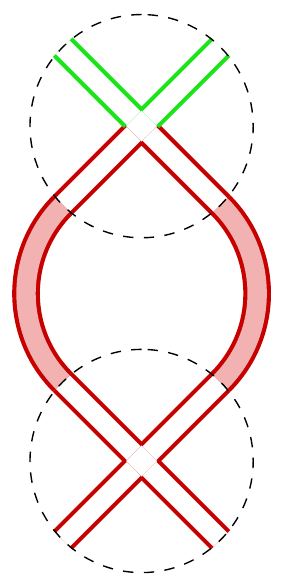}{412}\,,\hspace{.051\textwidth}
\includegraphicsd{.142}{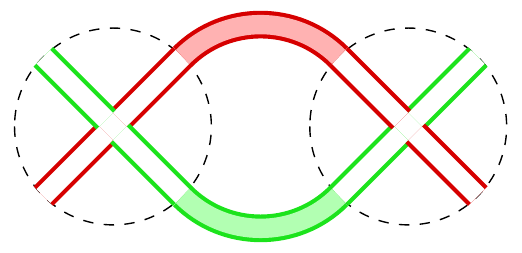} 
\,,\hspace{.051\textwidth}
\includegraphicsd{.142}{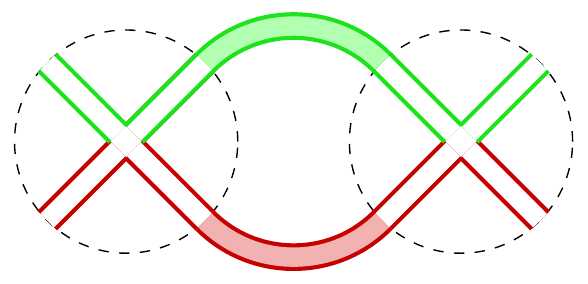}\,,\ldots \label{RadiaCorr}
\end{align}
``generate'' the \textit{effective vertex} $N \TrN(ABBA)$ (see below, how). 
Also disconnected vertices are generated; for instance, 
$\TrN(A)\times \Tr(A) $ is generated from $A^4$ (by contracting
non-consecutive half-edges) and from $ABAB$ (by contracting
the $B$ edges).

Therefore, the effective action should include these (and all corrections), and becomes\footnote{Notice that the $N$ can be re-absorbed in $Z$ and the bar-coupling constants, but other conventions are possible.} 
\begin{align}
\Gamma_N[A,B]& = \TrN  \Big\{ \overbrace{\frac {Z_A}2 A^2+\frac {Z_B}2 B^2 + \bar g_{A^4}\frac{1 }{4}A^4+ \bar g_{B^4}\frac{1}4 B^4+\frac12 \bar g_{ABAB} ABAB }^{\text{operators from the bare action (but with ``running couplings'')}} \\ & \qquad \,\,\,\,
+  \underbrace{\frac{1}{2}\bar g_{ABBA} ABBA + \frac12\bar g_{A|A}\TrN(A)\times A +\ldots}_{\text{radiative corrections}}  \Big\}
\end{align}

The effective vertices are obtained by taking the boundary graph of the 
radiative corrections. In other words, they are constructed from a Feynman graph---as were those in
\eqref{RadiaCorr} for the model \eqref{ABABaction}---as defined next, and 
explained with examples immediately thereafter.

\begin{definition}[Effective interaction vertex.]
Given a Feynman graph of a multi-trace multi-matrix model, first single out the 
traces $\TrN(U_1),\ldots,\TrN(U_r)$
that are not contracted by a propagator. Second, 
pick an arbitrary side of a ribbon-propagator and travel along the diagram 
with the orientation induced by the clockwise orientation of the interaction vertices,
listing in that order the letters that label the half-edges of these (cf. Ex. \ref{ex:orienta}) until one comes
back to the initial, chosen propagator (on the same side); 
call the thus obtained word $ w_1$. Repeat this process picking
an unvisited side of a propagator, and iterate until all ribbon propagators visited once by both sides (and thus 
all uncontracted half-edges are listed exactly once), say, at the $s$-th iteration.  The effective vertex $O_G^{\mathrm{eff}}$ of the graph $G$ is defined by 
\begin{align*}
O_G^{\mathrm{eff}}=\underbrace{\TrN (w_1)\times  \TrN (w_2)\times \cdots \times \TrN (w_s)}_{\text{from vertices
contracted with propagators}} \times \!\!\underbrace{
\TrN(U_1)\times   \TrN (U_2)\ldots\times \TrN(U_r)}_{\text{from vertices
uncontracted with propagators}}
\end{align*}
Since the words appear
inside the trace, the construction is evidently independent of the propagators we started with to construct each word $w_1,\ldots,w_s$. 
\end{definition}
\noindent 
 \begin{minipage}{.62\textwidth}
\begin{example}\label{ex:Nfactor}(Graphs containing an empty loop.) To illustrate the effective vertex construction of a two-matrix model,
consider the graph on the right, which corresponds to a 
correction from the operators
\begin{align}
O_1&= \bar g_1 \TrN(B^3ABA)\,, \\  O_2 &= \bar g_2 \TrN(A^2)  \TrN(A^3BAB)  \,.
\end{align}
The effective vertex is $N \bar g_1\bar g_2\TrN (A^2) \TrN(B^3A^2BA^2)$. 
The quadratic trace comes from the uncontracted trace in $O_2$;
the long word comes from the ``outward'' loop and the
factor $N=\TrN 1_N$ from the inner, empty word. 
\end{example}
 \end{minipage}\,
 \begin{minipage}{.36\textwidth} \hspace{1.0941cm}
\raisebox{-.56\height}{\includegraphics[width=22.5ex]{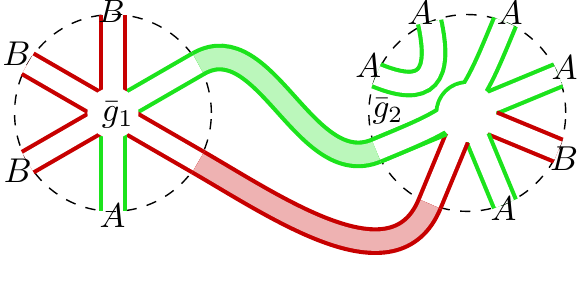}}
 \end{minipage}
\\[2ex]
\noindent 
 \begin{minipage}{.62\textwidth}
\begin{example}\label{ex:orienta}(Orientation of loops.) With the operators 
\begin{align}
O_1&= \bar g_1 \TrN( CFBDEA)\,, \\  O_2 &= \bar g_2 \TrN(  CDFABCEA) \,, 
\end{align}
we now illustrate the orientation of the loops. Each
operator endows the interaction vertex with an orientation.
The effective vertex should be read off respecting it.
This means that outward loops are clockwise oriented
and inward loops anti-clockwise. 
The effective vertex is $\bar g_1\bar g_2 \TrN( CFCE) \times \TrN(CDFADE )$.
\end{example}
 \end{minipage}\,
 \begin{minipage}{.36\textwidth} \hspace{1.0941cm}
\raisebox{-.56\height}{\includegraphics[width=23ex]{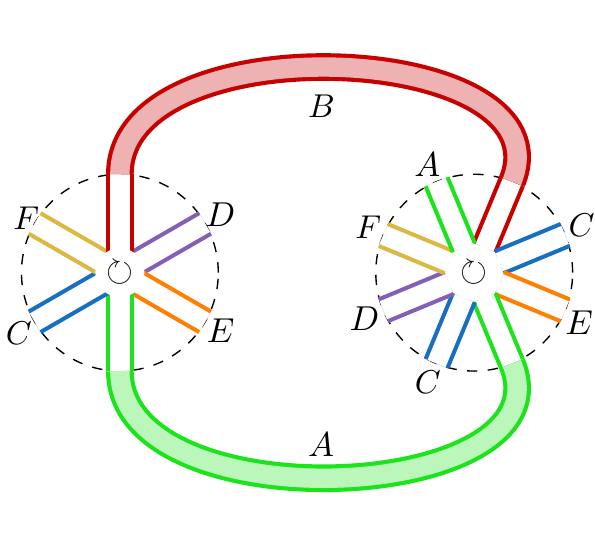}}
 \end{minipage}
\\[2ex]
\noindent 
\begin{minipage}{.62\textwidth}
\begin{example}(Propagators joining different traces in the same interaction vertex.) \label{ex:disconn_as_graph}Consider now the graph on the right. 
\begin{align}
O_1&= \bar g_1 \TrN(ADBADB)\,, \\  O_2 &= \bar g_2 \TrN(BACDBACD)\,, \\
O_3&= \bar g_3 \TrN(C^2)\times\TrN(D^3BDB) \,, \\ O_4 &= \bar g_4 \TrN(A^4) \times \TrN(D^6) \,,
\end{align}
(and possibly other more operators making the action real). 
The effective vertex is $\bar g_1\bar g_2 \bar g_3\bar g_4 \times \TrN(BDBD^7)\times \TrN (A^3DACDBACDADB)$. 
This graph is also a one-loop (see Definition \ref{def:oneloop} 
for the subtleties that appear in the presence of multi-trace interactions).
\end{example}
 \end{minipage}\,
 \begin{minipage}{.36\textwidth} \hspace{.41cm}
 \includegraphicsd{.79}{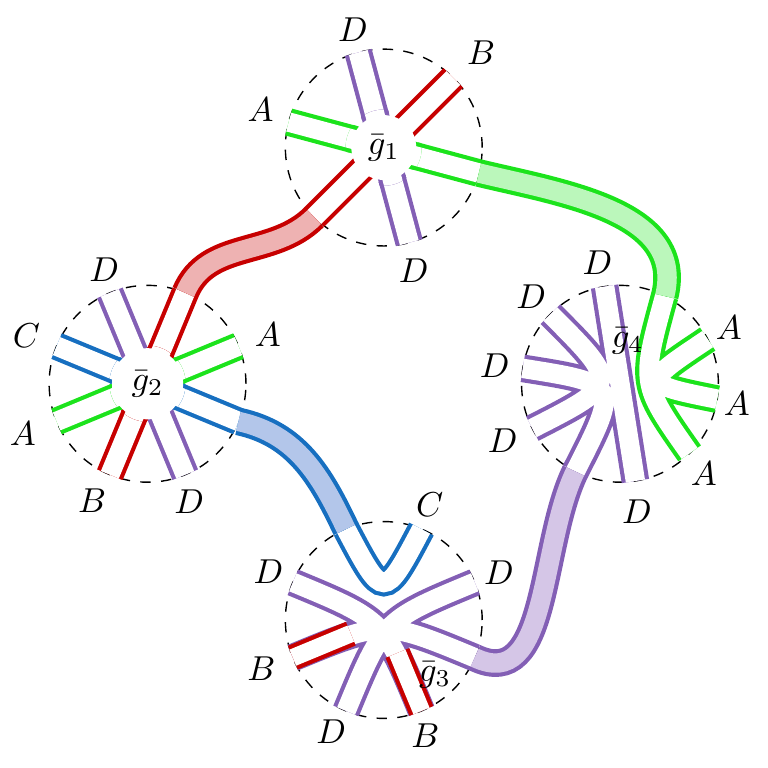}
 \end{minipage}\\[2ex]
%

In the presence of multi-traces, the one-loop condition cannot
be formulated purely in terms of the first Betti-number $b_1(G)$. Instead
\begin{definition}\label{def:oneloop}
Let $G$ be a ribbon graph of a multi-trace multi-matrix model. We denote by $G^{\circ}$ the 
one-dimensional skeleton obtained after collapsing the interaction vertices\footnote{In 
the single-trace random matrices literature these are sometimes called ``stars'' \cite{GuionnetFreeAn}. In
the definition of the one-dimensional skeleton it is implicit 
that we ignore discrete spaces obtained from the many traces that might be floating 
around the one-dimensional complex. This will be clear in Ex. \ref{ex:1loop}.} to points and the 
propagators (edges between interaction vertices) to ordinary edges. 
 A \textit{one-loop graph} of a multi-matrix model with multi-traces
 is a ribbon graph $G$ 
 whose skeleton $G^\circ$ 
 is one-particle irreducible (1PI; or, equivalently,  a $2$-edge connected graph)
 and which additionally has a first Betti-number $b_1(G^{\circ})= 1$. 
\end{definition}

\begin{example} \label{ex:1loop}The next three diagrams are all one-loop graphs: \\[.5ex]
\begin{align*}G_1= \raisebox{-.56\height}{\includegraphics[width=3cm]{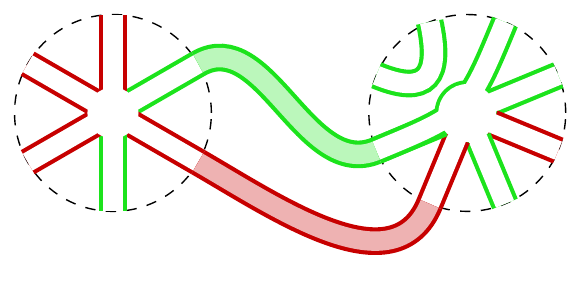}}\,,\,\, 
G_2=\raisebox{-.56\height}{\includegraphics[width=3cm]{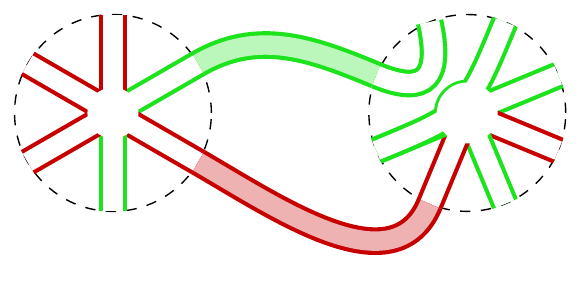}}\,,  \,\,
G_3=\raisebox{-.61\height}{\includegraphics[width=3cm]{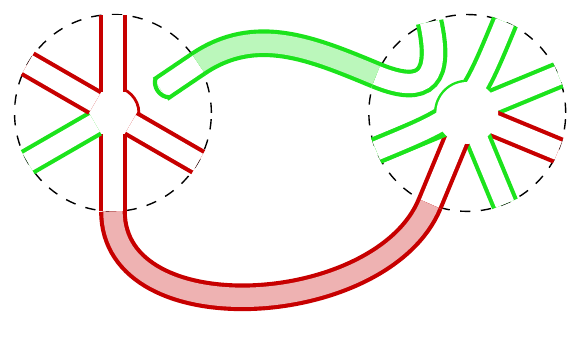}}\,.\end{align*}
(We omit the coupling constants $\bar g_i$ by now, since we care about 
topology in this example). First, $G_1$ has $b_1(G_1^\circ)=1$; next, although $b_1(G_2) \neq 1$,
since thinning the edges and collapsing the stars (dashed circles) 
to points yields a circle, $G_2^\circ$ does have first Betti-number 1. The 
same argument holds for $G_3$. Having these graphs explained the subtleties
of the multiple traces, we give now ordinary examples. Regarding\\[-2.5ex]
\begin{align*}G_4=\raisebox{-.41\height}{\includegraphics[width=1.4cm]{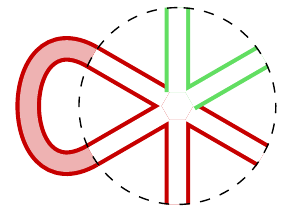}}\,,\, \qquad G_5=\raisebox{-.46\height}{\includegraphics[width=2.76cm]{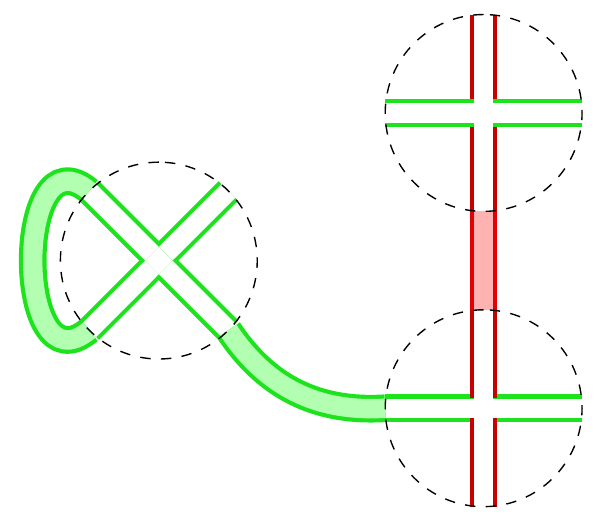}}\,, \qquad G_6=\raisebox{-.47\height}{\includegraphics[width=3.199cm]{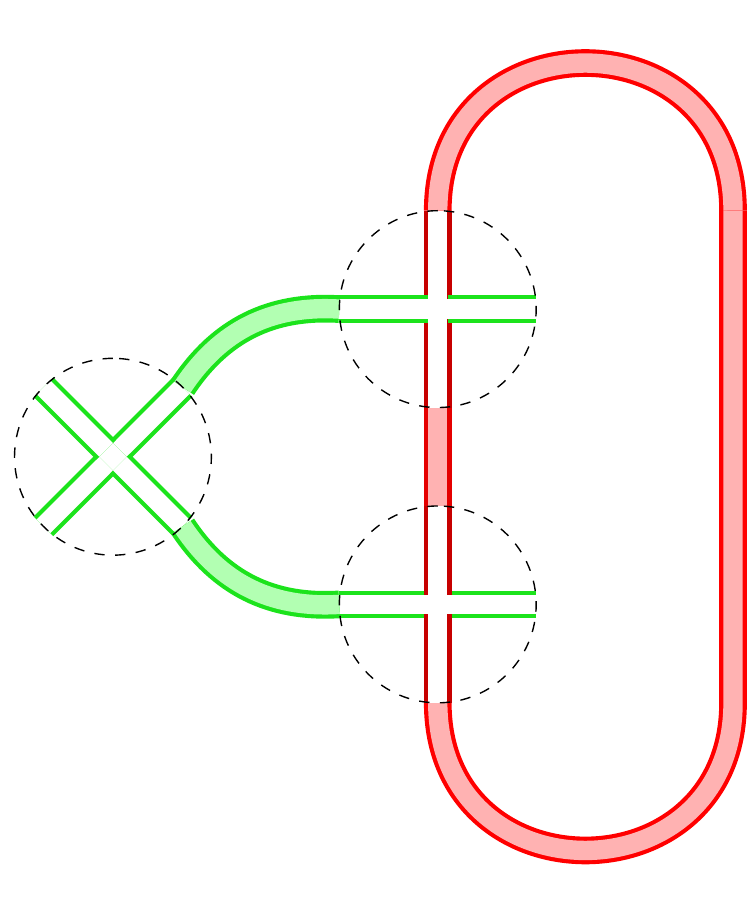}}\,\,,\end{align*}
only the \textit{tadpole} $G_4$
is a one-loop graph; although $b_1(G_5^\circ )=1$,
$G_5$ is not 1PI. And $G_6$ is such that $G_6^\circ$ has two loops, $b_1(G_6^\circ)=2$, so neither $G_5$ nor $G_6$ satisfy Definition \ref{def:oneloop}. \end{example}

\subsection{Including multi-traces and the main result}

We denote by $\Nmax \in \N$ the energy scale at which the bare 
action describes the system\footnote{ The reason for the notation $\Nmax$
is that, at the end, that integer can be thought of as being $\infty$. One computes
first all with finite $\Nmax$ and then takes the limit $\Nmax\to \infty$}. 
The renormalization flow modifies then the probability measure used to compute
observables as follows.
The starting point of the flow is the bare  action $S$ (or the measure $ \dif \mu\uvmax$ defined by it) \begin{equation}S:\H_{\Nmax}^n\to \re\qquad 
\dif \mu\uvmax=\exp\Big\{-S\big[\mathbb X\hp{\Nmax}\big]\Big\} \dif{\mathbb X}\Lebesgue\hp{\Nmax} \,.
                   \end{equation}
                   The ``$\textsc{\small{uv}}$'' in the measure emphasizes that 
the action $S$ that defines the probability measure $\dif\mu\uv=\exp\big\{-S[\mathbb X]\big\}(\dif \mathbb X) \Lebesgue$ is the bare action. 
                   In order to flow towards a lower energy scale $N< \Nmax$, 
                   a regulator $R_N$ takes care of integrating the \textit{higher modes} (i.e. the matrix
entries $N<i,j\leq \Nmax$ of each of the $n$ matrices; see Appendix \ref{app:RN}).
This smoothens the idea of step-by-step
integration \cite{Brezin:1992yc} of the $N+1$-th momentum shell,  in order to obtain from  
 ensembles of matrix of size $N+1$, effective ensembles of $N\times N$ matrices. This
 idea was put forward in \cite{EichhornKoslowskiFRG} for the one-matrix model
 in a quantum gravity context. Other renormalization theories based on Polchinski equation have been addressed in \cite{krajewskireiko}.  \par 
The system at that lower scale $N$ is described by the  
                     effective action $\Gamma_N$
                   and by the respective measure $\dif\mu \eff$ at the scale $N$,
                   \begin{align}\Gamma_N : \begin{cases}                                            
\H_{N}^n\!\!\!\!\!\!& \to \re, \\\mathbb X\hp{N}\!\!\!\!\!\!&\mapsto \Gamma_N[\mathbb X\hp N]
                                           \end{cases}
                   \qquad \dif\mu\eff(\mathbb X\hp N)= \exp\big\{-\Gamma_N [\mathbb X\hp{N} ]\big\} \dif{\mathbb X}\Lebesgue\hp{N}\,.
                   \end{align}

It can be rigorously proven  \cite{FRGEmultimatrix} that the effective action satisfies Wetterich equation,
\begin{equation}
\partial_t \Gamma_N[\mathbb{X}]= \frac12 \STr \Big(\frac{\partial_t R_N}{ \Hess \Gamma_N [\mathbb{X}]+R_N}\Big)\,,\label{Wetterich}
\end{equation}  but as pointed out 
in the introduction, this is not  the approach we follow in this article.
We rather assume that the renormalization flow is governed by an equation of the form
\eqref{Wetterich} and let ribbon graph theory dictate us 
the several objects that appear, specially the algebra obeyed by the Hessian. 
 If an expansion in $\mathrm U(N)$-invariant operators exist,
one is able to split the supertrace as follows:
\begin{equation}\label{rhsW}
 \frac12\STr \Big\{\frac{\partial_t R_N}{ \Hess \Gamma_N [\mathbb{X}]+R_N}\Big\} 
 = \sum_{k=0}^\infty   \underbrace{ \vphantom{\frac12} \bar h_k(N,\eta_1,\ldots,\eta_n )}_{\text{$R_N$-dependent part}} \times \underbrace{ \frac12(-1)^k \STr \big\{  (\Hess \Gamma \inter_N [\mathbb X])^{\star k} \big\}}_{\text{regulator-independent part}}\,,
\end{equation}
where $\bar h_k(N, \{\eta_1,\ldots,\eta_n\}) $ is a function of $N$ and the anomalous
dimensions $\eta_c=-\partial_ t \log Z_c$; 
finally, $\Gamma\inter_N[\mathbb X]$ is
the interaction part of $\Gamma_N$, which will be constructed below. 
Since we are looking for a ``universal'' algebras (not in 
the usual sense, but in the sense
that they will appear independent on the regulator $R_N$)
details on $R_N$ are placed in Appendix \ref{app:RN}.
\par 

In order to find the algebra $\mathscr A$
where the Hessian of the effective action lies,
let us search for the identity element of $\mathscr A$.
Because this algebra should contain $M_n(\Cn\otimes \Cn)$  (still seen as a a vector space),
we assume that $\mathscr A$ is also a matrix algebra of the form
 $\mathscr A=M_n(\A_n)$ for certain $\A_n$, and define the \textit{supertrace}\footnote{This is a historical 
 terminology which should not evoke supersymmetry.} $\STr$  
 on a matrix $\mathcal P=(P_{a,b})_{a,b=1,\ldots, n} \in \MA$, $P_{a,b} \in \A_n$ by
\begin{align}
\label{Supertrace}
\STr (\mathcal P) =\sum _{a=1} ^n   \Tr_{\A_n}(P_{a,a})\,
\end{align}
in terms of $\Tr_{\A_n}$, 
 where $\A_n$, its product $\star$ and its trace $\Tr_{\A_n}$  are to be determined.  

For this purpose, we observe that the effect of the kinetic terms, at a graph level, is just elongating the 
ribbons, and since all $R_N$-dependence has been absorbed in the coefficients
$\bar h_k$ in \eeqref{rhsW}, we conclude that the Hessian of the kinetic terms 
cannot modify the effective vertex at all: since, for $a,b,c,d\in\{1,\ldots,n\}$,
\begin{subequations}\label{identityonAn}
\begin{align} 
 \cdots \Hess O_{d,b}  \star \Hess \Big\{\frac12  \Tr (X_c^2)\Big\}_{b,a} \cdots & = \cdots \nonumber
\includegraphicsd{.31}{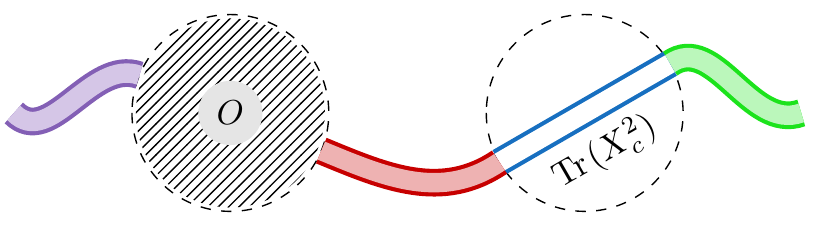} \cdots  \\
&= \cdots \includegraphicsd{.31}{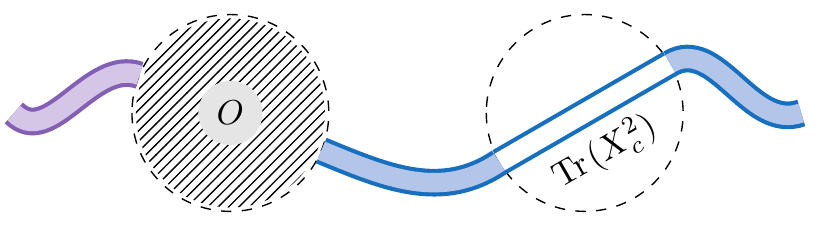}\times \delta_{a}^c  \delta_{b}^c   \cdots  \\
  \cdots\Hess \Big\{\frac12  \Tr (X_c^2)\Big\}_{a,b} \star  \Hess O_{b,d}  \cdots& = \cdots  \nonumber \includegraphicsd{.31}{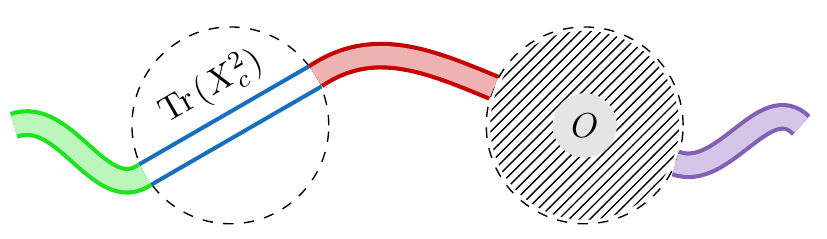}  \cdots \\
 &= \cdots  \includegraphicsd{.31}{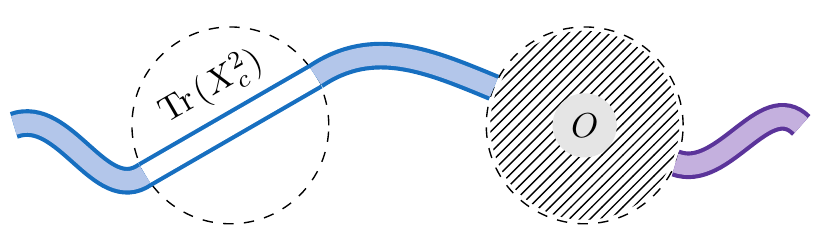} \times
 \delta_{a}^c  \delta_{b}^c \cdots 
\end{align}
\end{subequations}
for any interaction vertex $O$. 
On the other hand,
the double trace  terms $[\Tr X_c]^2$ ``cut'' the interaction vertex:
\begin{subequations}\label{NotQuiteidentityonAn}
\begin{align}  \cdots
\Hess \Big\{\frac12  [\Tr (X_c)]^2\Big\}_{a,b} \star  \Hess O_{b,d}\cdots & = \cdots\delta_{a}^c  \delta_{b}^c 
\includegraphicsd{.31}{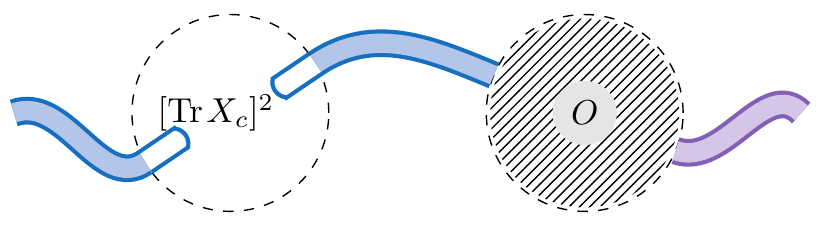}\cdots \\[1ex]
\cdots \Hess O_{d,a}  \star \Hess \Big\{\frac12  [\Tr (X_c)]^2\Big\}_{a,b}\cdots& = \cdots\delta_{a}^c  \delta_{b}^c  \includegraphicsd{.31}{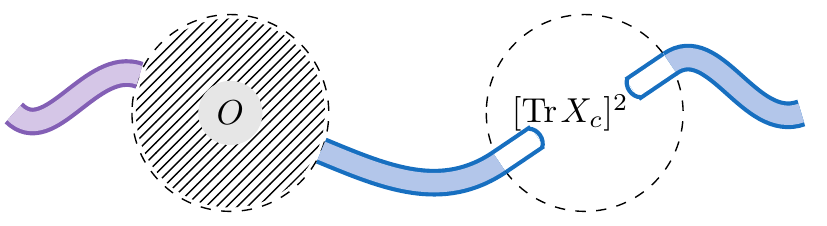}\cdots
\end{align}
\end{subequations}
By \eeqref{identityonAn}, $\Hess_{c,c} \frac12  \Tr (X_c^2)  =   1_N\otimes 1_N$ (no sum) is the left and right identity of $\A_{n,N}$, 
and by \eeqref{NotQuiteidentityonAn} there is another constant generator in  $\A_{n,N}$ that, by
the previous graph argument, is not proportional 
to $1_N\otimes 1_N$ (and therefore cannot be the identity) and which we denote by $1_N\totimes 1_N$.
\begin{definition}
We define  $\A_n:=\Cfree{n}^{\,\otimes\, 2} \oplus \Cfree{n}^{\,\totimes \,2}= [\Cn \otimes \Cn] \oplus [\Cn \totimes \Cn]$. Again,
 this is simplified notation for $\A_{n,N}$ defined
 as $\A_n$, but with $\CnN$ instead of $\Cn$. 
\end{definition}
 So far, $\A_n$ is only a vector space and $\totimes$ is just a symbol
 which will be different from $\otimes$ when we leave the 
 category of vector spaces and grasp $\A_n$ already as an algebra. 
 The bilinearity of $\totimes$  is due to 
 the coupling constants $\bar g $ of interaction vertices 
$O=\TrN[ \bar g Q_1 ]\TrN Q_2 = \TrN Q_1 \TrN [ \bar g  Q_2 ]$, which can ``enter into any 
trace''. Thus, $\totimes$ must satisfy $(\lambda U )\totimes W = U \totimes( \lambda W)
 $ for complex $\lambda$ and $U,W\in \Cn$.
The noncommutative Hessian can be extended
to products
of traces as follows: 
 \begin{definition}\label{def:HessondoubleTrs}
   On double traces $\Hess : \im \Tr ^{\otimes 2} \to M_n(\A_n )$ 
is given by
\begin{align}\label{Hess_bitrace}
\Hess \big\{\!\Tr ^{\otimes 2} (P\otimes Q)\big\}
=\Hess P \times \Tr Q + \Hess Q \times \Tr P + \Delta (P,Q)\,,
\end{align}
where $
                                               \Delta (P,Q)=(\Delta_{a,b}(P,Q))_{a,b=1,\ldots,n} $ has the following $M_n$-matrix entries  \begin{align}\Delta_{a,b}(P,Q)& =\partial_{X_a}\TrN  P \totimes \partial_{X_b}\TrN  Q
+ \partial_{X_a}\TrN  Q \totimes \partial_{X_b} \TrN P  \nonumber \\
& =\Day_{X_a}  P \totimes \Day_{X_b}  Q
+ \Day_{X_a}  Q \totimes \Day_{X_b}  P\,. \label{bitraces_entries}
\end{align}
 \end{definition}

 \begin{lemma}
  The trace $\Tr_{\A_n}$ on $\A_n$ is defined\footnote{Just as
 the operators $\Day_A$ and $\partial_A$, this abstract trace 
 is the result of matrix-trace calculations with entries.
 In \cite{FRGEmultimatrix} the relation to those is exposed.
 If the reader wants to look up there, there is however a change
 of notation; $\otimes_\tau$ there is $\otimes$ here; also $\otimes$ corresponds with  our $\totimes$.} in terms of $\TrN$ by linear extension of  
 \begin{subequations}\label{FRGtraces}
 \begin{align}
  \Tr_{\A_n} (P\otimes Q) &= \Tr_N^{\otimes 2} (P\otimes Q) = \TrN (P) \times \TrN(Q)\label{FRGtracesTensor}\\
  \label{FRGtracesOTensor} \Tr_{\A_n} (P\totimes Q) & = \TrN (PQ)
 \end{align}
 \end{subequations}%
 \end{lemma}
\begin{proof}
The tadpoles yield the desired relations. To obtain the first,
for any fixed $c\in \{1,\ldots,n\}$, 
consider an interaction vertex $O=\bar g \Tr (X_c P X_c Q)$
with $P,Q\in \Cn$ satisfying $\partial_{X_c}P=\partial_{X_c} Q=0$
(e.g. take $P,Q \in \C_{\langleb n-1 \rangleb}= \C \langleb X_1,\ldots, X_{c-1}, X_{c+1},\ldots, X_n\rangleb $). The contribution to the rhs of the 
flow equation is \begin{equation}
                  \frac12 \STr  \Hess O = 
 \frac{\bar g}2 \TrA (P\otimes Q+ Q\otimes P)  +
\text{terms not implying $c$-propagators}\,.\label{tadpolesintermediaA}
                 \end{equation}
The value of the first two summands is determined by the effective vertex of the graphs 
that the Hessian computes according to \eeqref{doubleNCder}. These are such that 
the two ribbons are attached at the only two $X_c$ matrices in $O$,  
\[\Hess_{c,c} O = \hspace{-7ex}
\includegraphicsd{.32}{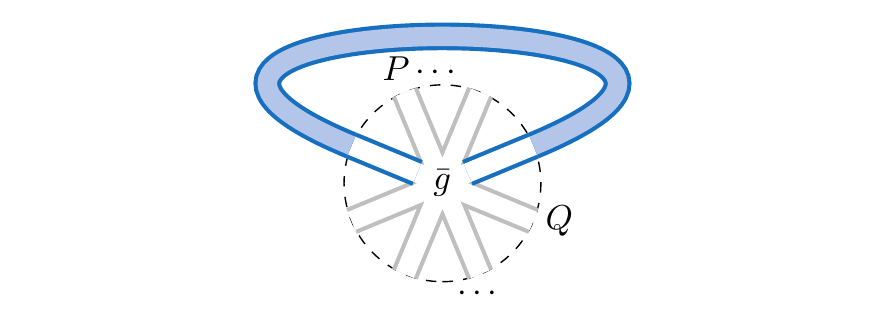}\hspace{-7ex} + 
\includegraphicsd{.198}{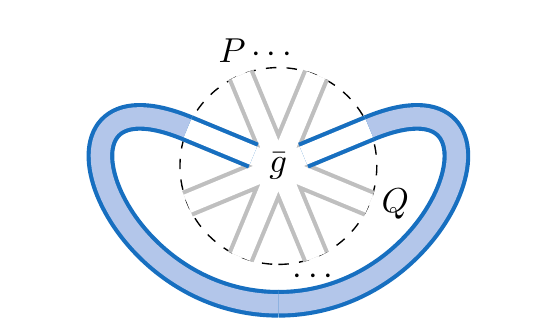} 
\]
    The ellipsis means that in the graphs, $P$ is the word after the 
    contracted $X_c$ running clockwise until the next $X_c$, after which $Q$
    begins. The seemingly different propagator contraction is just 
    an attempt to reflect that in the first graph $P$ is inside the loop 
    and $Q$ outside, with these words in the other way around for the second graph.
    However these two graphs are indistinguishable, thus, for 
    each graph the effective vertex  reads $\bar g\TrN(P) \times \TrN(Q)$,
    so by \eeqref{tadpolesintermediaA}, \eeqref{FRGtracesTensor} follows. To obtain the other
product, we consider tadpoles with the ends of the propagator 
on different traces of the same operator. Let
\[O'=\bar g' \TrN(P X_c) \TrN(Q X_c)= 
\includegraphicsd{.12}{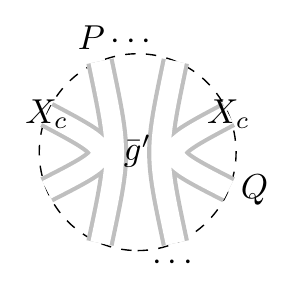}
\,
\Rightarrow\,
\Hess_{c,c} O' = \includegraphicsd{.18}{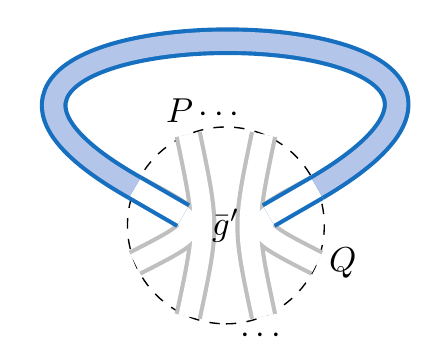}+
\includegraphicsd{.18}{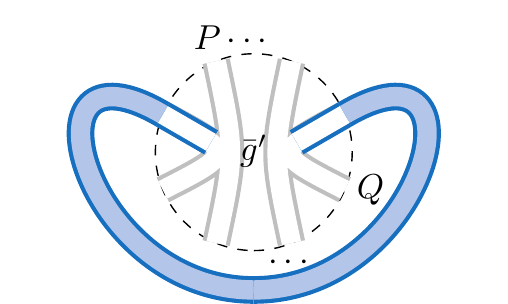}  \,.
\]
By \eeqref{Hess_bitrace} and \eeqref{bitraces_entries}, 
\begin{equation}
                  \frac12 \STr  \Hess O = 
 \frac{\bar g'}2 \TrA (P\totimes Q+ Q\totimes P)  +
\text{terms not implying $c$-propagators}\,.\label{tadpolesintermedia}
                 \end{equation}The effective vertex of each graph is $\TrN(PQ)$, which must be the value of 
$\TrA(P\totimes Q)$, but since the graphs are indistinguishable, also of $ \TrA(Q\totimes P)$,
Therefore, \eeqref{tadpolesintermedia} implies \eeqref{FRGtracesOTensor}.\par 
Now let us consider the general case, where $P$ might depend on $X_c$ (the dependence of $Q$ on $X_c$ can be likewise implemented, additionally, but the argument 
is the same in essence). Suppose that $P=P_{\text{\tiny L}} X_c
P_{\text{\tiny R}}$, where $P_{\text{\tiny L}},P_{\text{\tiny R}}
\in \Cfree{n}$ are monomials independent of $X_c$. In this simple case,
the rhs of Eq. \eqref{tadpolesintermediaA} 
receives the correction $\bar g \TrA [ P_{\text{\tiny L}} \otimes P_{\text{\tiny R}}X_c Q+
P_{\text{\tiny R}}X_c Q\otimes P_{\text{\tiny L}} ] $,
by the formula \eqref{doubleNCder} for the Hessian. However, since $P,Q$ are arbitrary, these terms cannot contribute 
to the coefficient of $\TrN P \TrN Q$ in $ \frac{1}{2} \STr (\Hess O)$,
since none of the graphs in such correction 
comply with having effective vertex (proportional to) $\TrN P \TrN Q$. Therefore
such contributions can be ignored. A similar treatment for 
a generic word $P$ and $Q$ that might contain $X_c$ concludes
also the proof of \eqref{tadpolesintermedia} without
restrictions on $P$ and $Q$ imposed above. 
\end{proof}

 \par 
In order to justify 
\eeqref{rhsW}, we now define both $\mathcal C$ and $\Gamma\inter _N [\mathbb X]$ by 
\begin{align}\label{correctC}
R_N+\Hess \Gamma_N  [\mathbb X] =:\mathcal C\inv + \Hess \Gamma\inter _N [\mathbb X]\,,
\end{align}
where $ \Gamma\inter _N [\mathbb X]$
 contains only interaction vertices (and $[\TrN X_c]^2$ counts as such; $\mathcal C$ is the correlation or 
 inverse propagator). That is, $ \Gamma\inter _N $ is defined in such a way 
that the Gaussian part $\dif \gamma\eff $ in the effective measure is factorized out:
\begin{subequations}
\begin{align}
\dif\mu\eff ( \mathbb{X} )&=  \ee^{- \Gamma_N [\mathbb X] } \dif \mathbb X\Lebesgue= \ee^{- \Gamma\inter _N [\mathbb X] } \dif\gamma\eff (\mathbb X)\,,
\\ 
\dif \gamma\eff (\mathbb X)&= \prod_{c=1}^n \ee^{- Z_c \TrN(X^2_c/2)} (\dif X_c)\Lebesgue\,.
\end{align}\label{splitmeasure}%
\end{subequations}%
Notice that one could have been tempted, inspired by \cite{FPinverse}, to separate the Hessian in its 
field-independent part (defined by its vanishing when $\mathbb X=0$)
and the field dependent part as performed in the 
functional renormalization treatment to one-matrix models by \cite{EichhornKoslowskiFRG}. The ``field part'' of the algebra $\A_{n,N}$ 
consists of non-trivial words 
(i.e. except multiples of $1_N\otimes 1_N$ and $1_N\totimes 1_N$). 
But the presence of double-trace 
quadratic operators $\frac12[\Tr (X_c)]^2 $, whose Hessian 
is $\frac12\Hess \{ (\Tr X_c)^2\} =   \diag_n [0,\ldots, 1_N\totimes 1_N ,0,\ldots, 0]$ 
with the non-zero in the $(cc)$-th entry of the $M_n$-block diagonal $\diag_n$, lies in the field-independent part,
and this impairs (as we see now) the Neumann expansion.
On the other hand, the definition \eqref{correctC} 
guarantees that the propagator $\mathcal C\inv$ is  $1_n \otimes 1_N\otimes1_N$ 
times a function (on $[1,\ldots, N]^2$), due to   
\begin{align*}\sum_{c=1}^n \Hess \Big\{\frac12  \Tr (X_c^2)\Big\}= \sum_{c=1}^n  \diag_n [0,\ldots, \underbrace {1_N \otimes 1_N }_{\text{$c$-th place}},0,\ldots, 0] = 1_n\otimes 1_N\otimes 1_N\,.\end{align*}
 When the wave function renormalization constant $Z_c$
is supposed to be equal for all matrices, $Z_c=Z$, then $\bar h_k(N,\eta)$, $\eta=-\partial_t \log Z$,
and the sums in $\bar h_k$ can be approximated by integrals of the form $\frac{1}{N^2}\int (\partial_t r_N)_{\sigma,\tau} \mathcal C^{k+1}_{\tau,\sigma}\dif \sigma\, \dif \tau $ that remain finite as $N\to \infty$.
We do not study the space of possible regulators  (in itself, interesting), 
but we stress that the expansion \eqref{rhsW} in unitary
invariants is an assumption. Ideally, as commented in \cite{EichhornKoslowskiFRG}, 
since $R_N$ breaks the symmetry, the  \eeqref{rhsW} should 
include operators $\STr(  \partial_t R_N \mathcal C  [ \Hess \Gamma_N\inter[\mathbb X] \mathcal C ]^{\star k})$.
However, identifying these operators with broken unitary symmetry is out of our present scope and for now 
the best one can do is to split, as in \eeqref{rhsW}, the rhs of Wetterich 
equation in $R_N$-dependent and $R_N$-independent part. The main result of this article is the unique description of the latter.
\begin{theorem}\label{thm:main}
For multiple-trace self-adjoint $n$-matrix ensembles,  
assume the rhs of Wetterich equation to be computable in terms of $\mathrm U(N)$-invariants as the geometric series \eqref{rhsW} in the Hessian.
Moreover, require that in \eeqref{rhsW}
only one-loop graphs are generated. Then the powers 
$(\Hess \Gamma \inter_N [\mathbb X])^{\star k} $
are taken in the algebra  
$ M_n(\mathcal A_{n,N},\star)$ of $n\times n$ matrices with entries in $\mathcal A_{n,N}$, explicitly
\begin{align}
 M_n(\mathcal A_{n,N}) = M_n(\C)\otimes \mathcal A_{n,N}\,,\quad 
 \mathcal A_{n,N}= \CnN^{\,\otimes \, 2} \oplus \CnN^{\,\totimes \, 2}\,,
\end{align}
whose product is given entry-wise by $(\mathcal P \star \mathcal Q)_{a,c} = \sum_{b=1}^n P_{a,b} \star Q_{b,c}$ for $\mathcal P=(P_{a,b})_{a,b=1,\ldots, n},$ and $ \mathcal 
 Q=(Q_{a,b})_{a,b=1,\ldots, n} \in M_n(\mathcal A_{n,N})$, and 
each entry $P_{a,b}$ and $Q_{b,c}$ obeys the following multiplication rule,
given here on homogeneous elements of $\mathcal A_{n,N}$: %
for any $P,Q,U,W \in \Cn$,%
\begin{subequations}\label{FRGalgebra}%
\begin{align}
(U \otimes W) \star ( P\otimes Q)  &=  PU \otimes WQ\label{FRGalgebraIII} \,, \\
(U\totimes W) \star ( P\otimes Q)  &=U \totimes PWQ  \label{FRGalgebraII} \,, \\
(U \otimes W) \star ( P\totimes Q)  &= WPU \totimes Q \,,\label{FRGalgebraI} \\
(U\totimes W) \star ( P\totimes Q) &= \TrN (WP) U\totimes Q \label{FRGalgebraIV}\,.
\end{align}\end{subequations} %
\end{theorem}
\begin{proof}
Section \ref{sec:proof} is the proof. 
\end{proof}

In other words, if one computes functional renormalization of matrix models
with a product different from 
\eeqref{FRGalgebra}, \textit{either} contributions that do not have the one-loop structure 
appear in the $\beta$-functions \eqref{betaeta}, \textit{or} it is impossible to 
compute the renormalization flow by splitting, in regulator-dependent and regulator-independent parts
as in \eqref{rhsW}---regardless of what $\bar h_k$ might be.

\begin{remark} \label{rem:questions}
There are two interesting limiting cases\footnote{I thank R\u azvan Gur\u au for questions that motivated this remark (which gives partial answers).}, large-$N$ (together with the 
initial scale of the bare action $N_\infty \to \infty$) and small-$N$.
From Figure \ref{fig:ideal} it is evident that $N$-factors 
appear only when one-loop graphs have the
``empty word'' $1_N$ at any side. This suggests that the algebra 
of Theorem \ref{thm:main} could be reduced to \eeqref{FRGalgebraIII}, 
but actually double-traces appear again in \eeqref{FRGalgebraIV}, and 
$\TrN(Q_1)\times \TrN(Q_2)$ compete with terms of the form $N\TrN(P)$.
Further, this argument should be thoroughly investigated, since 
the ensemble in the large-$N$ depends also on the power-counting, that is, 
on the solution for the 
$\kappa_\alpha$ and $\lambda_\alpha$; see the discussion just above \eeqref{betaeta}. 
For $\beta$-functions computed with the algebra \eeqref{FRGalgebra}, see \cite[Thm. 7.2]{FRGEmultimatrix} and \cite{komentarzABAB}.
The critical behavior could be explored in the sense of \cite{EichhornKoslowskiFRG}
as eigenvalues of the stability matrix, namely $-\mathrm{Eig}\big\{(\partial \beta_\alpha (\eta^{\!\!\balita\!\!}, \{g^{\!\!\balita\!\!}\}/\partial g_{\alpha'} )\big\}_{\alpha,\alpha'}$, where the bullet means the fixed-point solutions  of the system \eqref{betaeta}, $\beta_{\alpha}(\eta^{\!\!\balita\!\!}, \{g^{\!\!\balita\!\!}\})= 0$ and $\eta_{c}(\eta^{\!\!\balita\!\!}, \{g^{\!\!\balita\!\!}\})=0$ for all 
interactions $\alpha$ and all matrices $c=1,\ldots,n$. In the large-$N$, 
for the two-matrix model with 48 operators (that is the number of operators in a sextic truncation) compatible with the symmetries of the $ABAB$-model, the unique fixed point solution with a single positive eigenvalue of the stability matrix happens when two coupling constants have the value $0.07972$ 
($1/4\pi = 0.07957...$ is the critical value for the coupling constants in \cite{KazakovABAB},
when one takes their sign and normalization conventions) and 
some double-trace operators like $\Tr^2_N(A),\TrN^2(A^2)$, $\TrN(A)\times\TrN(B^3)$, do contribute to the flow (at least
so with the regulator of App. \ref{app:RN}). The limit $N\to 1$  ($t\to 0$) should yield the full effective action (see limits in App. \ref{app:RN}), but this is unexplored here and needs an independent study.   
In the worst of the cases, the full algebra \eqref{FRGalgebra} is needed
to next-to-leading-order or \textsc{nlo} corrections, but bounds on those \textsc{nlo}-terms are precisely
the beginning of an analytic approach. 
\end{remark}

\begin{remark}[The product $\star$ in terms of matrix entries]
Consider the permutation $\tau=(13)\in \mathrm{Sym}(4)$
and denote by $\mathrm{id}      $ the identity of the symmetric group $ \mathrm{Sym}(4)$.
Let $\rho,\pi \in \{ \mathrm{id}, \tau\}$.
Then, if $a,b,c,d=1,\ldots,N$, and $Y_1,Y_2,Y_3,Y_4\in \Cfree{n}\subset \MN$ are monomials, 
the four products of Theorem \ref{thm:main} are summarized in the following equation,
where the sum over $x,y=1,\ldots ,N$ is implicit: \begin{align}
(Y_1 \otimes_\rho Y_2) \star (Y_3\otimes_\pi Y_4)_{ab;cd}=
(Y_1)_{\rho(a)\rho(b)} (Y_2)_{\rho(x)\rho(y)} (Y_3)_{\pi(y)\pi(x)} (Y_4)_{\pi(c)\pi(d)}
\end{align}
where for $\varpi\in \{ \mathrm{id}, \tau\}$, $\otimes_\varpi=\otimes$ if $\varpi=\tau=(13)$ and 
$\otimes_\varpi=\otimes_{\mathrm{id}}=\boxtimes$ if $\varpi$ is the trivial permutation.
Also $\rho$ acts as element of $\mathrm{Sym}(a,b,x,y)$ and $\pi $ on $\mathrm{Sym}(y,x,c,d)$.
For instance, in the nontrivial case $\rho=\tau$, $\tau(a,b,x,y)=(x,b,a,y)$.
We remark that in order to keep the Hessian 
simple in this paper, the convention is the opposite of \cite{FRGEmultimatrix},
i.e. $\otimes_\tau$ there is $\otimes$ here; and the $\otimes$ of \cite{FRGEmultimatrix} corresponds 
with the $\totimes$ here. The particular permutation $\tau=(13)$ might seem at first
arbitrary, but it is actually natural and
can be found in \textit{op.cit.} or in \cite[Eq. 5]{GuionnetFreeAn}.
\end{remark}

\section{The proof of the main statement}\label{sec:proof} 
Figure \ref{fig:topoofproof} gives the logic structure in the proof.
By ``$s \subset \Hess_{a,b} ( O)$'' we abbreviate that $s$ is a summand in $\Hess_{a,b} (O)$. Further,
 $M,L,P,Q,R,S,T,U,V,W\in \Cn$ are arbitrary monomials.

\begin{proof*}{Proof of Theorem \ref{thm:main}}
Start with the $k$-th power of a Hessian. 
First, we argue that we can simplify
this situation and deduce the behavior 
regarding the $k$-th power for any $k$ from the square of a Hessian. 
Supertraces of products of Hessians will be sums over terms of the 
following form:
\begin{equation} \label{HessianwithGraph}
  \Hess_{  a,b}{(O_1) } \star
  \Hess_{  b,c}{ (O_2 )}\star 
  \Hess_{  c,d}{ (O_3) }\star \ldots\star \Hess_{ *, a} { (O_k) } \supset
 \includegraphicsd{.30}{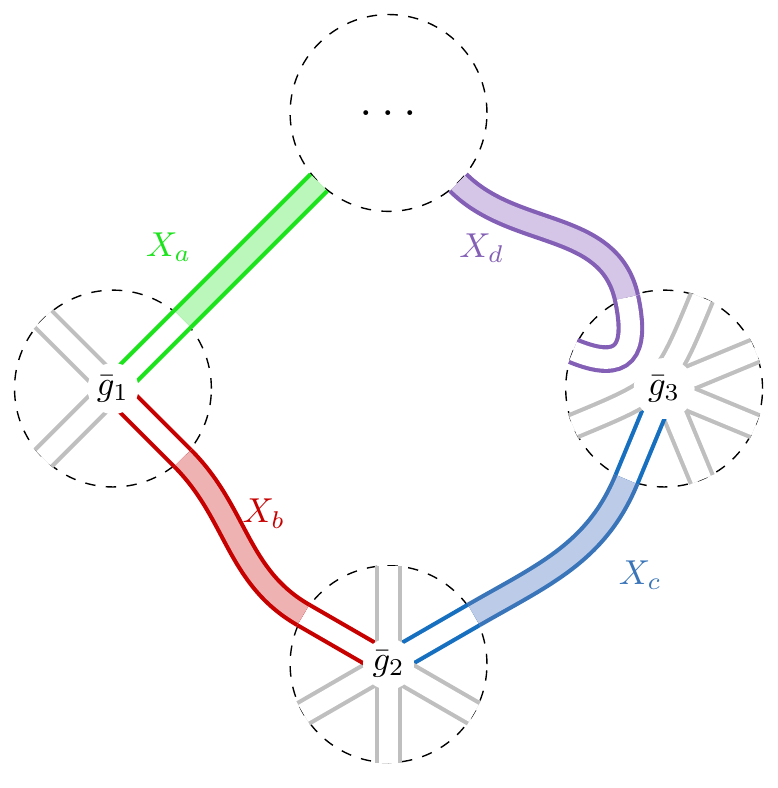}
 \end{equation}

The associativity of the product $\star$ follows from the definition
of effective vertices (but should be verified purely algebraically, 
after the product is constructed):
\begin{equation}
 \label{graphAssoc}
   \includegraphicsd{.30}{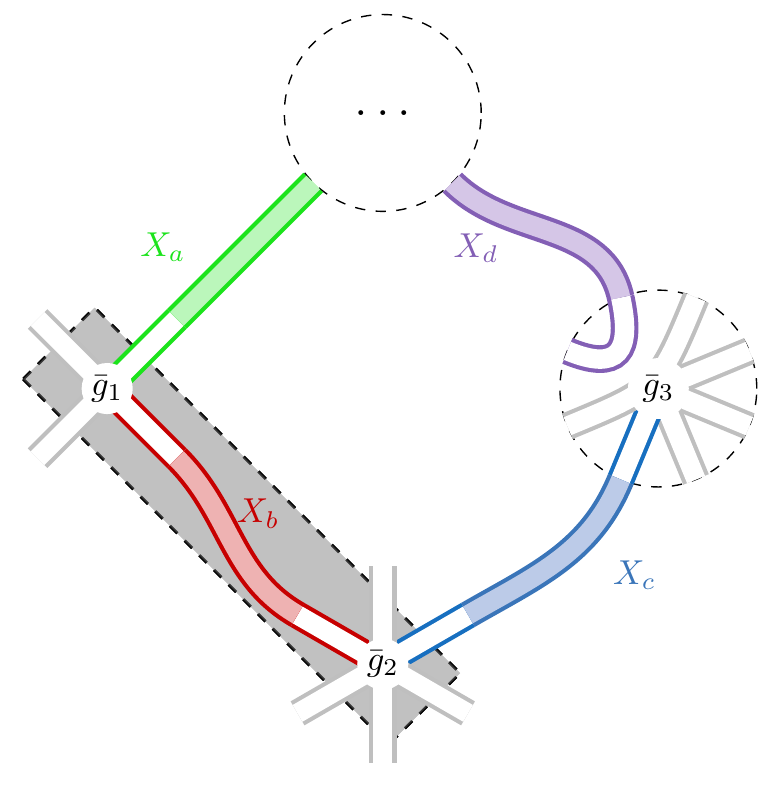}
\,\,\text{ \Large $=$ }
 \,\, \includegraphicsd{.30}{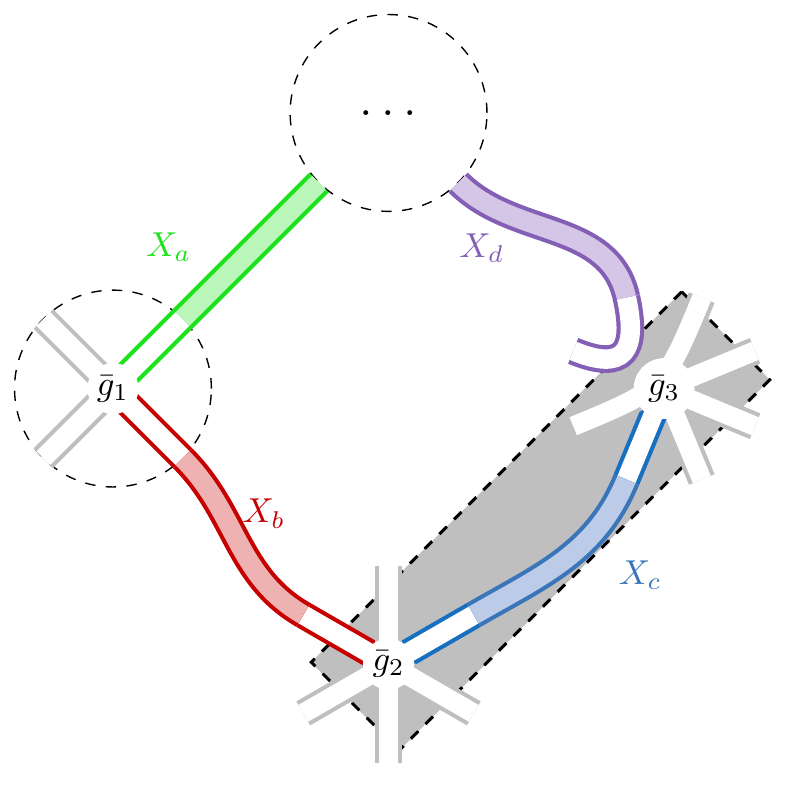}
\end{equation}
where the gray boxes with uncontracted, protruding ribbon edges
mean the new interactions formed from the two grouped interaction vertices.
The new cyclic order is determined by the propagator, together with the half-edges it is attached to, being shrunk.
The left corresponds to the 
$[\Hess_{a,b} (O_1) \star \Hess_{b,c} (O_2) ] \star \Hess_{c,d}(O_3) $ 
bracketing while the right one to 
$ \Hess_{a,b} (O_1) \star [ \Hess_{b,c} (O_2)  \star \Hess_{c,d}(O_3)]$.  \par
We have four cases, depending on the way
the four propagators in the loop connect the
interaction vertices of $k=2$ interaction vertices.
 The fact that $\mathscr A=M_n(\A_n)$ is an associative algebra
 (or recursive application of \eqref{graphAssoc}) allows us not to consider more cases.
 However, to determine the product, $k=3,4$ will yield
also useful information too.
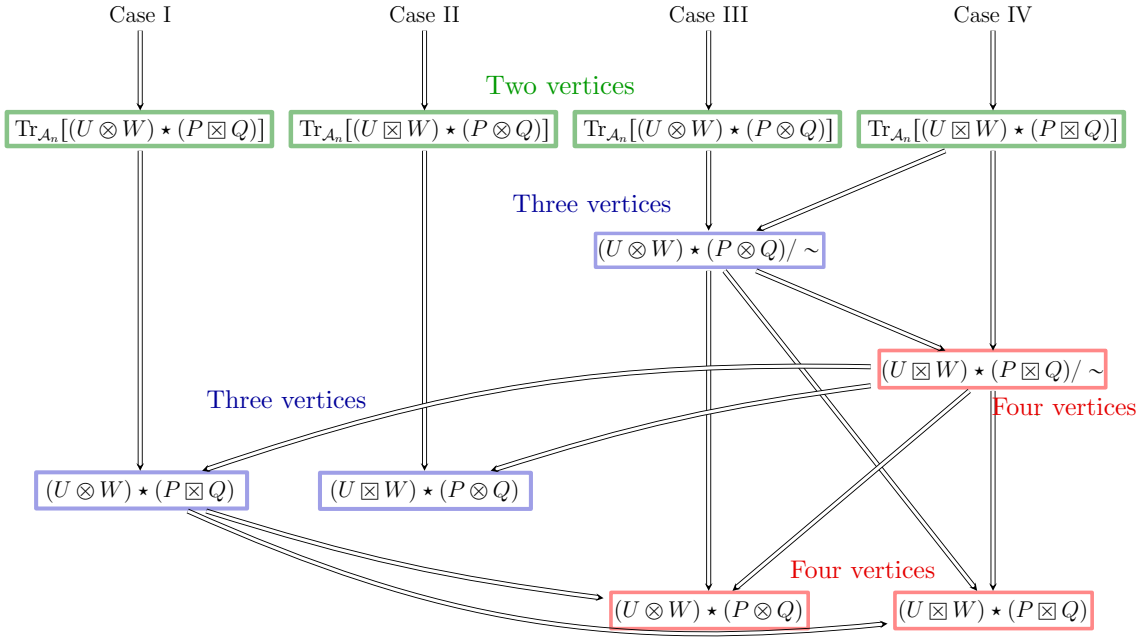
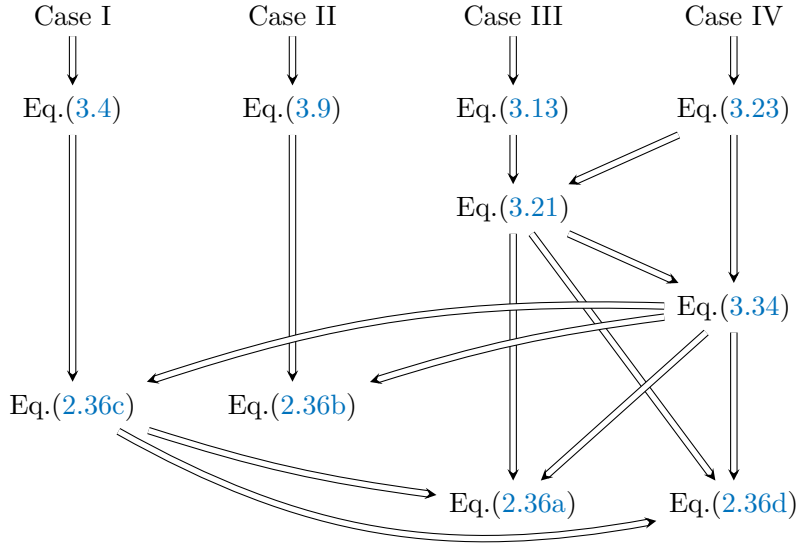
\begin{figure} 
\centering 
     \begin{subfigure}[b]{0.70\textwidth}
    \hspace{-2cm}
 \includegraphics[width=1.29\textwidth]{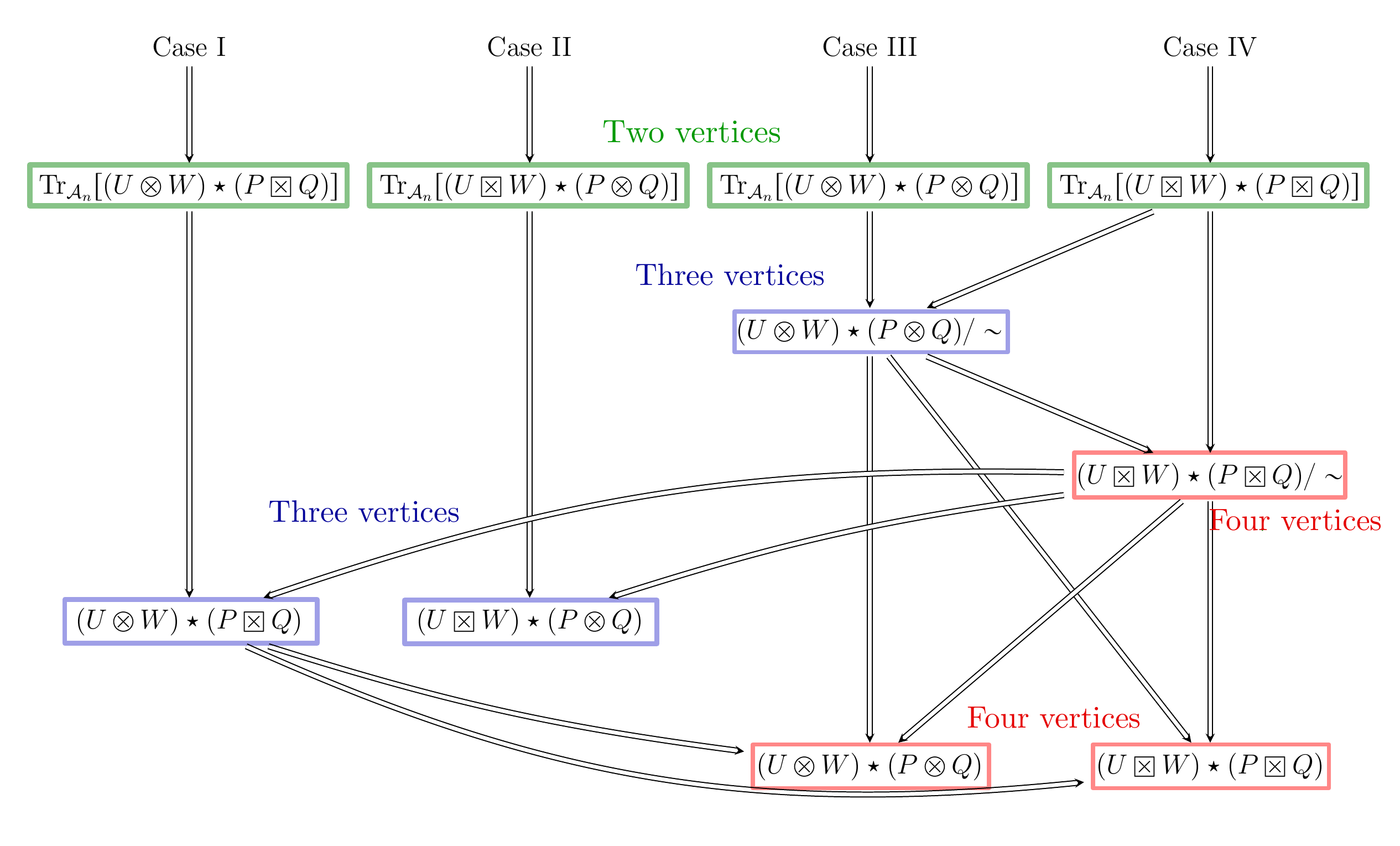}
 \caption{In the first row (top to bottom), we have the four different cases
 in one-loop diagrams in multi-trace matrix models. From each,
 we can recognize a product type in the lhs of \eeqref{FRGalgebra}.
 Starting from each case one can, autonomously, determine the value for the trace $\TrA$ of each 
 of these products using two interaction vertices. As one wants to determine the product itself, more 
 information (which turns out to be delivered by looking
 at other cases) is needed:
 For instance, knowing both traces for cases III and IV, one can determine the product in case III 
 up to the transformation  $X\otimes Y + V\totimes Z \mapsto \widetilde{X\otimes Y} + \widetilde{ V\totimes Z} = Y\otimes X + Z\totimes V$ (which is marked with ``modulo $\sim$'') by using a third interaction vertex, et cetera, as determined by this diagram. Observe that, although we combine the cases, 
 and present each column at once, 
 there is no (oriented) loop.\label{fig:logic_diagram}}
\end{subfigure}%
 \\[3.5ex] \centering %
     \begin{subfigure}[b]{0.7\textwidth}\centering    %
\begin{tikzcd}[column sep=normal ]  
\text{Case I} \ar[d, double, double distance=2pt,, -stealth]  & \text{Case II}\ar[d, double, double distance=2pt,, -stealth] & \text{Case III} \ar[d, double, double distance=2pt,, -stealth] & \text{Case IV} \ar[d, double, double distance=2pt,, -stealth] \\
\texteqref{CITr} \ar[ddd,double, double distance=2pt,-stealth] 
& \ar[ddd,double, double distance=2pt,-stealth] \texteqref{CIITr}
&\texteqref{CIIITr} \ar[d,double, double distance=2pt,-stealth]&  \ar[dd,double, double distance=2pt,-stealth]\texteqref{CIVTr} \ar[ld,double, double distance=2pt,-stealth]\\ 
& &   \texteqref{CIIISym} \ar[dr,double, double distance=2pt,-stealth] \ar[ddd,double, double distance=2pt,, , -stealth]  \ar[dddr,double, double distance=2pt,, -stealth]&  \\ 
& &   &  \texteqref{CIVSym} \ar[dll,double, double distance=2pt,-stealth,bend right=5] \ar[dlll,bend right=10,double, double distance=2pt,-stealth] \ar[dd,double, double distance=2pt,, , -stealth]  \ar[ddl,double, double distance=2pt,, , -stealth] \\ \texteqref{FRGalgebraI} \ar[drrr,bend right=20,double, double distance=2pt,-stealth]  \ar[drr,bend right=5,double, double distance=2pt,-stealth] & \texteqref{FRGalgebraII}  & &  \\
& & \texteqref{FRGalgebraIII} &  \texteqref{FRGalgebraIV} 
\end{tikzcd}  \caption{The equation numbers
corresponding with Figure \ref{fig:logic_diagram}.\label{fig:logic_diagram_withEqs}}
\end{subfigure}
\caption{The ``topology'' of the proof of Theorem \ref{thm:main} showing the 
absence of logic loops, notwithstanding the mix of cases in the proof. The 
arrows are implications. These diagrams show 
how we ``bootstrap'' the algebra.  \label{fig:topoofproof}}
\end{figure}

\begin{itemize}\setlength\itemsep{.4em}%

 \itemb \textsc{Case I:} \textit{When two ribbons in the loop lie in the same 
 trace in the first interaction vertex, but in different traces in the second:}

\begin{align}\raisebox{-.45\height}{ \includegraphics[width=4.3cm]{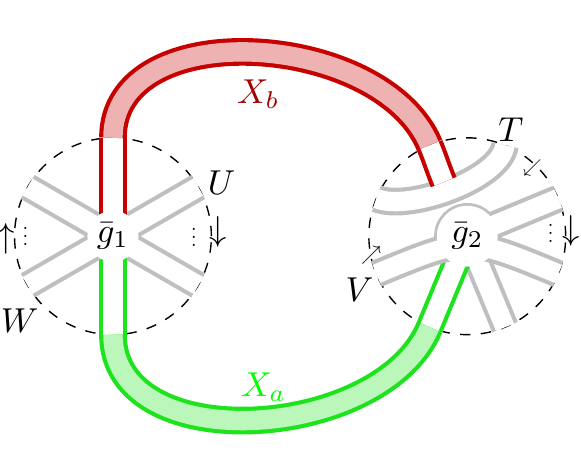}}\subset \Hess_{a,b} O_1 \star \Hess_{b,a} O_2\label{CaseI}
\end{align}
 
 Suppose that the interaction vertices  have one and two traces, respectively. 
  In fact they might have more traces, but these not being implied
 in the loop for the present case, they remain intact; thus, we do not loose generality 
 by this simplification. There exist then words $T, U,V ,W$ (which might be empty) such that  
\begin{align*}O_1 = \bar g_1 \Tr (UX_a W X_b )\quad \and \quad O_2=\bar g_2 \Tr (  X_b  T) \Tr (  X_a  V) \,.\end{align*}
 The words $T,U,V,W\in \Cn$ might contain 
 the letters $X_a, X_b$, but we are analyzing only the summand in the 
 lhs of \eqref{CaseI}.  
 To compute the contribution of the two Hessians to this precise
 summand we get by \eeqref{doubleNCder} $ \Hess_{a,b} O_1 \supset \partial_{X_a} \circ 
 \partial_{X_b } \Tr (U X_a W X_b) = U\otimes W$,
 and by \eeqref{bitraces_entries}, \begin{align*} \Hess_{b,a} O_2 \supset\Day_{X_b} \Tr (  X_b  T) \totimes  \Day_{X_a} \Tr (  X_a  V) = T\totimes V \,. \end{align*}  Now, since the effective vertex of  \eqref{CaseI} is formed by  shrinking the green and red propagators and merging the rest of the 
 ribbon half-edges while preserving the order, the graph \eqref{CaseI} implies that the effective vertex is $\bar g_1 \bar g_2\TrN(W T U V ) $. By Wetterich equation, 
 \begin{align} \label{CITr}
 \Tr_{\A_n}[ (U\otimes W) \star (V\totimes T) ] = \TrN(W T U V ) \,.
 \end{align}
 Since the lhs is a single trace, this is enough to conclude that the result of $ (U\otimes W) \star (V\totimes T)$ must be ``a $\totimes$ inserted somewhere
 in the cyclic word $ W TU V$'', otherwise it would be a product
 of the form $w_1\otimes w_2$ which, when traced, would yield a $N$-factor,
 in case that any of the words $w_1$ or $w_2$ is trivial, and
 a double trace if both are not trivial.
 We also know that the result of $(U\otimes W) \star (V\totimes T)$ must be an ordinary and not a cyclic word; thus, so far, we need to know how to root it, i.e. the expression for $(U\otimes W) \star (V\totimes T)$ 
 should be listed in 
\begin{align}
\quad 1& \totimes  WTUV ,& W& \totimes  TUV,& WT& \totimes  UV,& WTU& \totimes  V,& WTUV& \totimes  1\,, \nonumber \\
1& \totimes  TUVW,& T& \totimes  UVW,& TU& \totimes  VW,&TUV& \totimes  W,& TUVW & \totimes  1\,, \label{possib_prods_I} \\
1& \totimes  UVWT,& U& \totimes  VWT,& UV& \totimes  WT,&UVW& \totimes  T,& UVWT& \totimes  1\,, \nonumber\\
1& \totimes  VWTU,& V& \totimes  WTU,& VW& \totimes  TU,&VWT& \totimes  U,& VWTU& \totimes  1\,. \nonumber
\end{align} 
 To discard 
the wrong ones, we first  consider the following interaction vertices: 
\begin{align*}
O_1&= \bar g_1 \TrN (X_bW X_a U)\,, \\ O_2&= \bar g_2 \TrN (X_b T) \TrN( X_c V) \,,\\
 O_3&=\bar g_3   \TrN (X_c R) \TrN(S X_a)\,. 
\end{align*}
and the corresponding product of Hessians of each of these (in that order), which 
contains in particular, the next graph:
     \begin{align*}\raisebox
    {-.45\height}{
 \includegraphics[width=4cm]{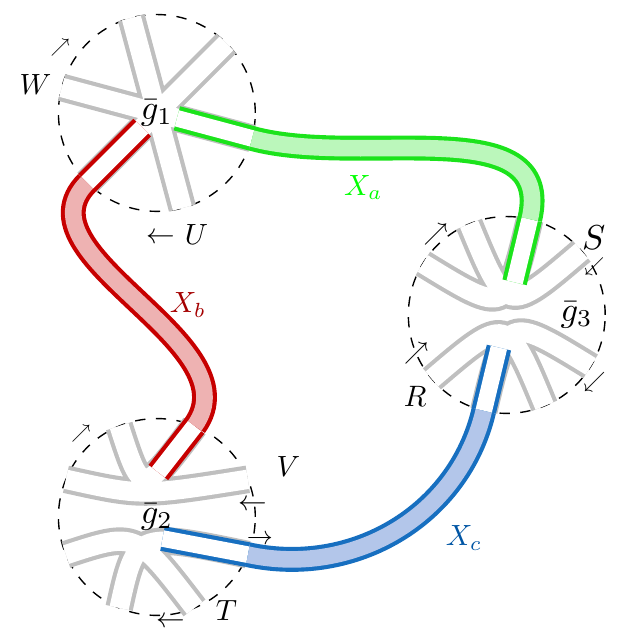} }\subset \Hess_{a,b} O_1 \star 
 \Hess_{b,c} O_2 \star 
 \Hess_{c,a} O_3
 \end{align*} 
The effective vertex must be $\TrN(VWSU) \TrN(RT)$, thus 
\begin{equation}
 \label{TraceMussI}
\TrA \big\{ [( W\otimes U ) \star (V\totimes T)] \star (R\totimes S)\big\} = \TrN(VWSU) \times \TrN(RT)
\end{equation} 
One can use the previous graph to discard elements in the list \eqref{possib_prods_I}.
For instance, we suppose that  $( W\otimes U ) \star (V\totimes T) = W \totimes  TUV$.
For the product inside curly 
brackets $\{ \ldots\}$,  using \eeqref{IVc} or \eeqref{IVd}, (equivalently, \eeqref{CIVSym}; \ref{fig:logic_diagram_withEqs}) one gets the following 
possibilities:
\begin{equation}
 =\begin{cases}
  S\totimes W \TrN(RTUV) &  \text{if \eeqref{IVc} holds}\,, \\
   W\totimes S \TrN(RTUV) & \text{if \eeqref{IVd} holds}\,.
  \end{cases}
\end{equation}
But the trace of it yields in either case $\TrN(SW) \TrN(RTUV)$
which differs from \eeqref{TraceMussI}. Thus 
$( W\otimes U ) \star (V\totimes T) = W \totimes  TUV$ is impossible.
By the same token, with the same counterexample above, one discards 
the possibilities that do not contain a factor of the empty word $1$, except $( W\otimes U ) \star (V\totimes T) =UVW\totimes T$.

Regarding those possibilities containing the factor of $1$, following any of the prescription 
of the  leftmost column in \eqref{possib_prods_I} for the square brackets product, and the Case IV,
which is to say either \eeqref{IVc} or \eeqref{IVd},
for the resulting multiplication of the form $w_1\totimes w_2 \star w_3\totimes w_4 $, one easily sees that these
generate a factor $\TrN(S)$; likewise, those possible products on the rightmost columns \eqref{possib_prods_I}
generate a factor $\TrN (R)$. Both lead then to contradiction with the previous graph. 
Therefore indeed $( W\otimes U ) \star (V\totimes T) =UVW\totimes T$, i.e. \eeqref{FRGalgebraI} holds.
\par

  \itemb \textsc{Case II:}  \textit{When two ribbons in the loop lie in the same 
 trace in the first interaction vertex, but in different traces in the second:}
\begin{align}\label{CaseIIgraph}
 \includegraphicsd{.26}{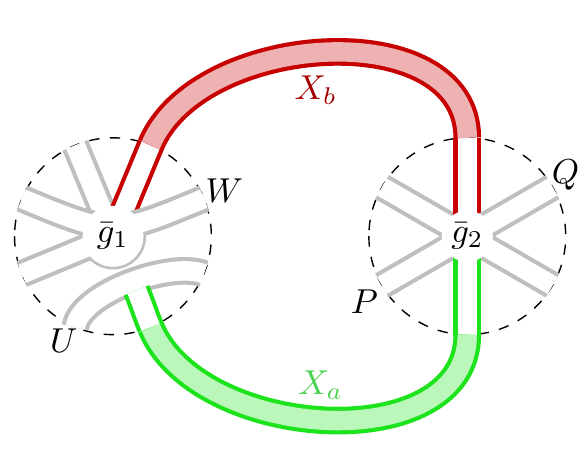} \subset 
 \Hess_{a,b} O_1 \star 
 \Hess_{b,a} O_2  
 \end{align} 
This case is proven by swapping the roles of the first and second
interaction vertices in Case I. Since the proof is analogous,
we rather sketch it. Take the next operators:
\begin{align*}O_1= \bar g_1 \TrN(X_aU)
\Tr(X_b W)\,, \quad O_2= \bar g_2 \TrN(X_a P X_b Q)\,. \end{align*}
Since the trace over $\A_n$ must coincide with the effective vertex to \eqref{CaseIIgraph}, \begin{align}
\TrA\big( U\totimes W \star P\otimes Q  \big) =
\TrN(QUPW)                               \label{CIITr}
                              \end{align} Again, 
we have the following possibilities for the value of the element of $\A_n$ inside
the trace:
\begin{align}
1& \totimes  UPWQ ,& U& \totimes  PWQ,& \cdots&  & UPWQ& \totimes  1\,, \nonumber \\
 1& \totimes QUPW , &  Q&\totimes UPW ,  & \cdots & &   QUPW& \totimes 1\,,  \label{productchoicesII} \\ 
 & \vdots & &\ddots  && &\vdots \nonumber \\ 
1& \totimes  PWQU ,& U& \totimes  PWQ,& \cdots&  & UPWQ& \totimes  1\,. \nonumber
\end{align} 
Call $ U\totimes W \star P\otimes Q= J(W,T,U,V) \totimes K(W,T,U,V)$  the correct product listed here. 
By Case IV's partial conclusions, to wit \eeqref{IVc} and \eeqref{IVd}, one has 
\begin{align} J\boxtimes K \star R \totimes S  =\begin{cases} \TrN(K R) J\totimes S \\
\TrN(K R) S \totimes J \label{algebraicII}
\end{cases} \stackrel{\TrA}\mapsto \TrN(JS) \TrN(KR)\,.
\end{align}   
Now consider three vertices,
\begin{align*}\qquad
O_1= \bar g_1 \TrN(X_aU)
\Tr(X_b W)\,, O_2= \bar g_2 \TrN(X_c P X_b Q)\,, O_3=\bar g_3 \TrN (X_c R ) \TrN(S X_a) \,, 
\end{align*}
and the product of the Hessian applied to these.
By looking at the graph, 
\begin{align*}\includegraphicsd{.2}{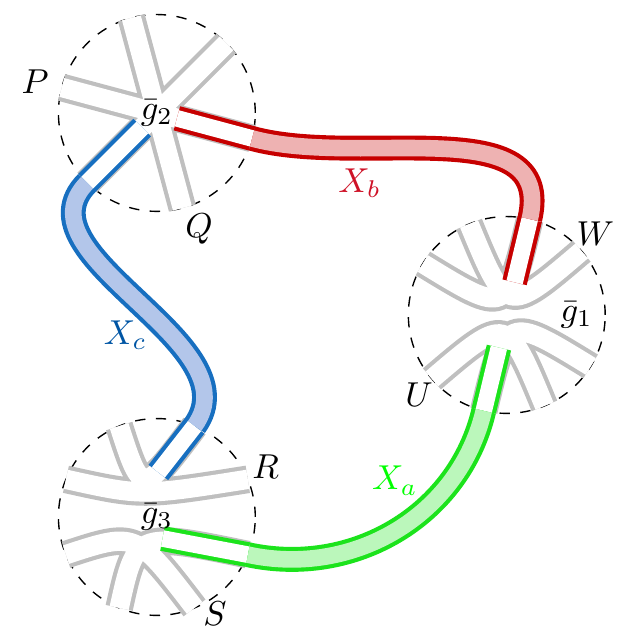}
\subset 
\Hess_{a,b} O_1 \star 
 \Hess_{b,c} O_2 \star 
 \Hess_{c,a} O_3\,,
\end{align*}
one deduces that $\TrA [  U\totimes W \star P\otimes Q \star R \totimes S]=\TrN (SU) \TrN(PWQR) $.
Since this holds for each $W,S,Q,R,P,U\in \Cn$, comparing with \eeqref{algebraicII},
we obtain that  $J=U$ and $K= PWQ$ is the right choice among \eqref{productchoicesII},
which means that $U\totimes W \star P\otimes Q=U \totimes PWQ$ and we have 
proven \eeqref{FRGalgebraII}.

 \itemb \textsc{Case III:}\textit{ When two ribbons in the loop lie on the same 
 trace in both the first and second interaction vertices:}

\begin{align}\raisebox{-.65\height}{
 \includegraphics[width=4cm]{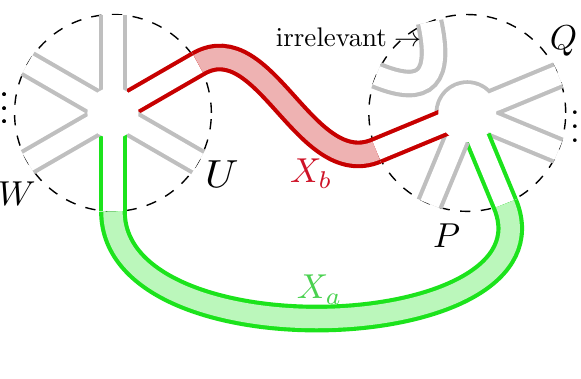}} \,\,\subset \,\,\Hess_{a,b} O_1 \star 
 \Hess_{b,a} O_2 \,.\label{CaseIV}
 \end{align}   We consider operators
$O_1= \bar g_1 \TrN (X_aW X_b U)  $ and $O_2= \bar g_2 \TrN (X_a P X_b Q) $.
These might have more traces, but as depicted above,  these being outside the loop,
do not suffer any transformation (in that summand) and can be ignored. Then  
$\Hess_{a,b} O_1 \star \Hess_{b,a} O_2 = (U\otimes W) \star (P\otimes Q)$.
According to Wetterich equation, the effective vertex
must be 
\begin{equation} \label{CIIITr}
\Tr_{\A_n}[(U\otimes W) \star (P\otimes Q)] =\TrN (PU) \TrN (WQ).   
\end{equation}
which implies either of the following possibilities:
\begin{subequations}\label{possib_prods_III}
\begin{align}
(U\otimes W) \star (P\otimes Q)& = \TrN(PU) W\totimes Q \label{IIIa} \\
(U\otimes W) \star (P\otimes Q)& = \TrN(PU) Q\totimes W \label{IIIb}  \\
(U\otimes W) \star (P\otimes Q)& = PU \otimes QW   \label{IIIc} \\
(U\otimes W) \star (P\otimes Q)& = UP \otimes QW \label{IIId}\\
(U\otimes W) \star (P\otimes Q)& = UP \otimes WQ \label{IIIe}\\  
(U\otimes W) \star (P\otimes Q)& = PU \otimes WQ\label{IIIf} 
\end{align}
\end{subequations}
To obtain the right one(s), we consider now the third power
of the Hessian, but in the contraction with the additional vertex
Case IV shall be here useful. By contradiction to each of the cases, we suppose that \eeqref{IIIa} holds.
Then consider the following interaction vertices:
\begin{align}
O_1&= \bar g_1 \TrN (X_aW X_b U)\,, \\ O_2&= \bar g_2 \TrN (X_b Q X_c P) \,,\\
 O_3&=\bar g_3   \TrN (X_c R) \TrN(S X_a)\,. 
\end{align}
By a similar ribbon graph argument,
we obtain, using the hypothesis,
that $(\Hess_{a,b} O_1 \star \Hess_{b,c} O_2 )\star \Hess_{c,a} O_3 $ which is, modulo the coupling constants $
[(U\otimes W) \star (P\otimes Q)]\star (R\totimes S)= [\TrN(PU) W\totimes Q]
\star (R\totimes S) $. Applying $\Tr_{\A_n}$ to this quantity we deduce, 
according to the partial conclusion in Case IV (recalling that the 
bracketing is irrelevant due to \eeqref{graphAssoc}), that
\begin{align}
 \qquad\Tr_{\A_n}\big\{ (U\otimes W) \star (P\otimes Q) \star (R\totimes S)\big\}  & = \TrN(PU) \Tr_{\A_n}\big\{
 (W\totimes Q ) \star  (R\totimes S)\big\} \\ & =\TrN(PU) \TrN(QR) \TrN(WS).
\end{align}
However, by looking at the graph that the product 
in curly brackets represents, the previous equation
cannot be true, for the graph leads to a single trace,
namely, $\TrN(WQRPUS)$. This is a 
contradiction with the supposition that 
$(U\otimes W) \star (P\otimes Q) = \TrN(PU) W\totimes Q$.
Hence, we discard  \eeqref{IIIa}. 
By the same argument in number of traces, we discard
also \eeqref{IIIb}. Further, with the next counterexample 
\begin{equation}
\label{CounterGraphIII}
\includegraphicsd{.22}{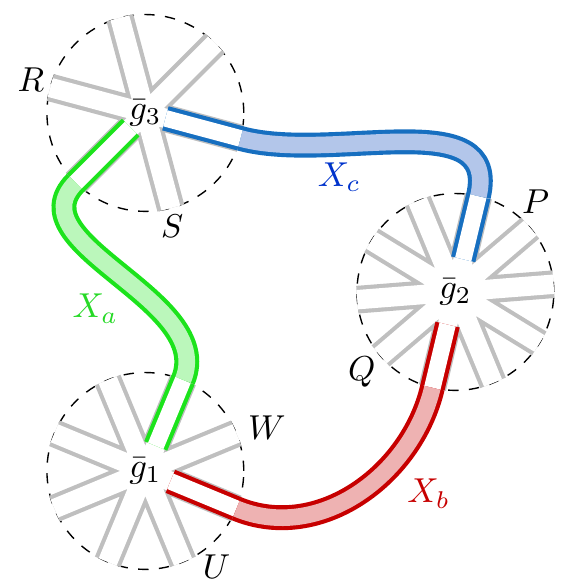}\subset  
 \Hess_{a,b} O_1 \star 
 \Hess_{b,c} O_2 \star \Hess_{c,a} O_3 \,,
\end{equation}
obtained from the Hessians of the operators
\begin{align*} 
O_1=  \bar g_1 \TrN (X_aW X_b U)\,, \,\, O_2 =  \bar g_2  \TrN (X_b Q X_c P)\,,\,\,
 O_3= \bar g_3 \TrN (X_a R X_c S )\, . 
\end{align*}
 
Then two further possibilities are ruled out, for, on the one hand, \eeqref{IIIc} implies that the effective vertex is $\TrN(RPU)\TrN(SQW)$;
and \eeqref{IIIe}, on the other hand,
implies that it is $\TrN (UPR)\TrN(WQS)$  (mod coupling constants).
Either is different from the effective vertex for the graph \eqref{CounterGraphIII}, namely 
 $\TrN( WQS )\TrN(RPU)$. So only the next two are possible:
\begin{align}
(U\otimes W) \star (P\otimes Q)  = UP \otimes QW  \qquad   
(U\otimes W) \star (P\otimes Q)  = PU \otimes WQ\,. \label{CIIISym} 
\end{align}

 We solve now Case IV and then determine which of the two is the right expression.

   \itemb \textsc{Case IV:}  \textit{When two ribbons in the loop lie on 
   different traces in both interaction vertices:}
 \begin{align}\label{CaseIV}\raisebox{-.45\height}{ 
 \includegraphics[width=4.0cm]{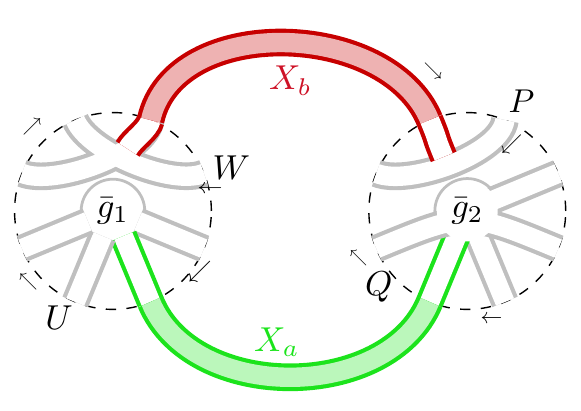}}\subset \Hess_{a,b} O_1 \star \Hess_{b,a} O_2\end{align}
 Notice that, notwithstanding the disconnectedness of the ribbons 
 in this graph $G_{\textsc{iv}}$,
 as a graph in the field theory context, what matters
 is the connectivity of its skeleton $G_{\textsc{iv}}^\circ$ ($\bar g_i$
 are the coupling constants for both traces inside the dashed
 circle, cf. Def \ref{def:oneloop}). We construct now operators  
 that yield the desired product. Let $O_1= \bar g_1 \TrN(UX_a) \TrN(X_b W) $
 and $O_2= \bar g_2 \TrN(QX_a) \TrN(X_b P) $. Then
 the product of Hessians in \eeqref{CaseIV} contains
  $ (U\totimes W) \star ( P\totimes Q)$ as a summand. 
  The effective vertex must be 
  what we obtain by shrinking the propagators. In turn, in the rhs of   
  Wetterich Equation \eqref{Wetterich} this effective vertex
  is obtained by tracing over\footnote{Actually one 
  has to trace over $M_n (\A_n)$, i.e. take the supertrace.
  But the trace corresponding to the $M_n$-block matrix  
  was already taken in \eeqref{CaseIV}. } $\A_n$, so 
\begin{align}
 \label{CIVTr}
  \Tr_{\A_n}[(U\totimes W) \star ( P\totimes Q)] =
  \TrN ( WP ) \times \TrN (UQ  ) \,.
\end{align}

 This means that the quantity in square brackets
 must be either of the following product formulas for $\star$: 
 \begin{subequations}
 \begin{align}
 (U\totimes W) \star ( P\totimes Q)  & =   P\totimes W \TrN(QU) \label{IVa} \\
 (U\totimes W) \star ( P\totimes Q)  & =  W\totimes P \TrN(QU)  \label{IVb} \\
  (U\totimes W) \star ( P\totimes Q)  & =  Q\totimes U \TrN(PW)  \label{IVc} \\
   (U\totimes W) \star ( P\totimes Q)  & =   U\totimes Q \TrN(PW)   \label{IVd} \\
    (U\totimes W) \star ( P\totimes Q)  & = PW \otimes QU   \label{IVe} \\
     (U\totimes W) \star ( P\totimes Q)  & = WP \otimes QU   \label{IVf} \\
      (U\totimes W) \star ( P\totimes Q)  & =  WP \otimes UQ   \label{IVg} \\
         (U\totimes W) \star ( P\totimes Q)  & = PW \otimes UQ   \label{IVh} 
 \end{align}
\end{subequations}
To determine the correct product, we consider higher powers of
 the Hessian, and one of the partial conclusion of the Case III, \eeqref{CIIITr}.\par 
By contradiction, suppose that \eeqref{IVa} holds. 
Then applying twice this equation,
\begin{align}
[(U\totimes W) \star ( P\totimes Q) ] \star (R\totimes S)& = \TrN(QU)  P\totimes W \star R\totimes S \\
& =  \TrN(QU) \Tr(SP) R \totimes W 
\end{align}
which when is traced in $\A_n$ yields 
\begin{align}
\Tr_{\A_n} \big\{ [(U\totimes W) \star ( P\totimes Q) ] \star (R\totimes S) \big\} =\TrN(QU ) \TrN(RW) \TrN(SP)\,.  \label{ContraIVa}
\end{align}
However, if we pick the next observables, 
\begin{align}
 O_1& = \bar g_1 \TrN (UX_a ) \TrN (WX_b )\\
 O_2& = \bar g_2 \TrN (PX_b ) \TrN (QX_c )\\
 O_3& = \bar g_3 \TrN (RX_c ) \TrN (SX_a )
\end{align}
the effective vertex for the summand
\begin{align*}\raisebox{-.45\height}{ 
\includegraphics[width=.2\textwidth]{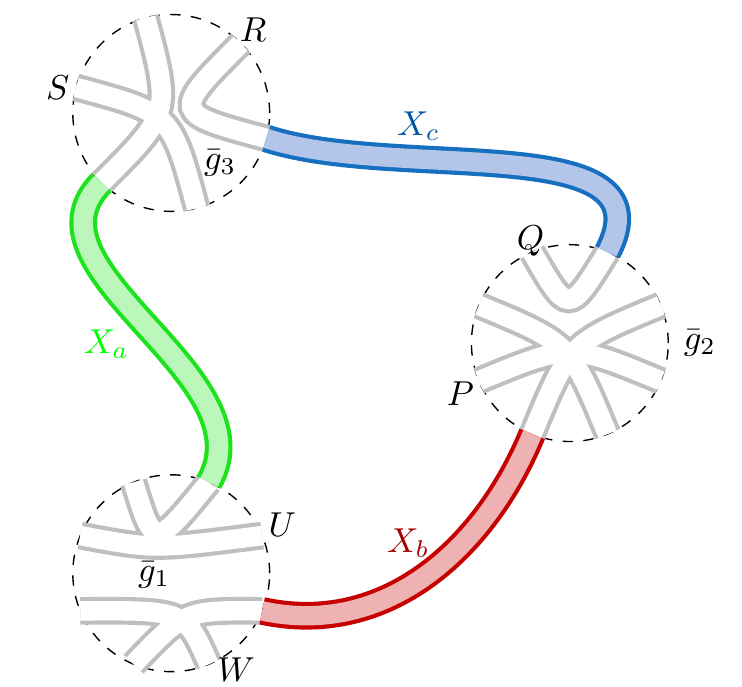}}\subset  
 \Hess_{a,b} O_1 \star 
 \Hess_{b,c} O_2 \star \Hess_{c,a} O_3 \,,
\end{align*} 
must be $\TrN(SU ) \TrN(PW) \TrN(QR)$, 
which is a contradiction with 
\eeqref{ContraIVa}. Thus \eeqref{IVa} is impossible. 
By the same token, one sees that  \eeqref{IVb} leads to 
an effective vertex $\TrN(PR ) \TrN(QU) \TrN(SW)$, which differs from 
$\TrN(SU ) \TrN(PW) \TrN(QR)$. Thus \eeqref{IVb} is not the right product either.
\par 
To rule out further products, we go to fourth degree in the Hessian.
Suppose that \eeqref{IVe} holds. Then 
\begin{align}
[U\totimes W \star P\totimes Q]\star[ T\totimes V \star M\totimes L ]
=&PW\otimes QU \star MV \otimes LT \\
=& \begin{cases}
    PWMV \otimes LTQU \\
    MVPW \otimes QULT
   \end{cases} \label{ContraIVe}
\end{align}
where the last equality lists the possibilities \eeqref{IIId} or \eeqref{IIIf}.
In either case, \eeqref{IVe} holds, then 
the trace of \eeqref{ContraIVe} reads $\TrN(PWMV) \times \TrN(LTQU) $. Again,
if we consider the operators 
\begin{subequations}\label{IVContra4operators}
\begin{align}
 O_1& = \bar g_1 \TrN (UX_a ) \TrN (WX_b ) &
 O_2& = \bar g_2 \TrN (PX_b ) \TrN (QX_c )\\
 O_3& = \bar g_3 \TrN (TX_c ) \TrN (VX_d ) &
  O_4& = \bar g_4 \TrN (MX_c ) \TrN (LX_a )\,.
\end{align}
\end{subequations}
we get from the summand
\begin{align*}
\includegraphicsd{.28}{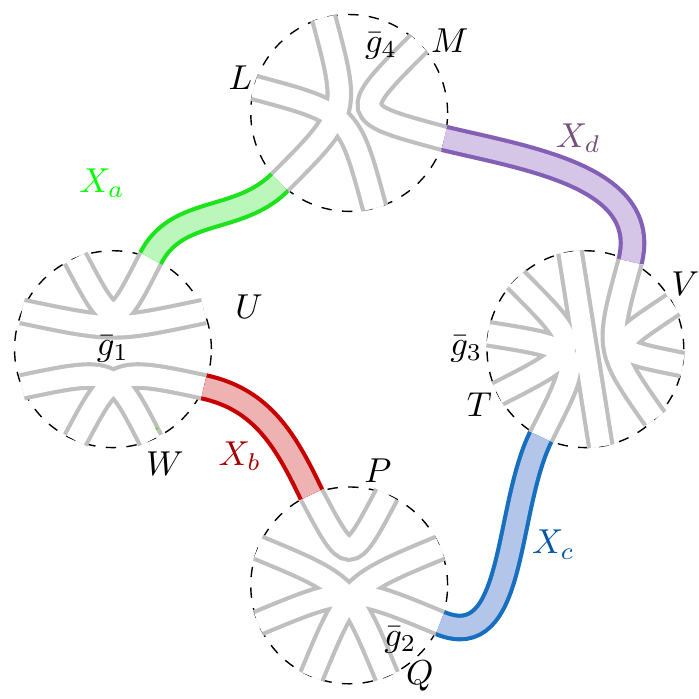}\subset  \Hess_{a,b} O_1 \star 
 \Hess_{b,c} O_2 \star \Hess_{c,d} O_3\star \Hess_{d,a} O_4
\end{align*}
  in the fourth power of the Hessian 
  the effective vertex $\TrN(LU) \TrN(WP) \TrN(QT) \TrN(VM)$. Since
  not even the number of traces coincides, \eeqref{IVe} is impossible.
  By the same trace-counting argument, the same operators \eqref{IVContra4operators}
  serve as a counterexample for the products \eeqref{IVf}, \eeqref{IVg} and \eeqref{IVh}.
  This leaves us only with possibilities \eeqref{IVc} and \eeqref{IVd}:
  \begin{align}
   \label{CIVSym}
     (U\totimes W) \star ( P\totimes Q)  & =  \begin{cases}
                                            Q \totimes U \TrN(PW) 
                                            \\ U\totimes Q \TrN(PW)   \end{cases}
  \end{align}

\end{itemize}

To finish the proof, we have to determine which of 
are the correct products,
consider the operators
\begin{align}
O_1&=\bar g_1 \TrN(X_a V X_b T ) \,,\\
O_2&=\bar g_2 \TrN(X_b W X_c U ) \,, \\
O_3&=\bar g_3 \TrN(X_cP  ) \TrN(X_d Q) \,,\\
O_4&=\bar g_4 \TrN(X_d R ) \TrN( X_a S ) \,, 
\end{align}
and the product of their Hessian (entries) 
\begin{align*}
(T\otimes V \star U\otimes  W) \star (P\totimes Q \star R \totimes S)
\end{align*}
This expression is given by 
\begin{align}&=
  \begin{cases} 
      TU\otimes WV \vphantom{\TrN}&  \text{if \eeqref{IVc} holds}\\
        UT\otimes VW\vphantom{\TrN}  &  \text{if \eeqref{IVd} holds}                                     \end{cases} \Bigg\}\star
        \begin{cases} 
      S\otimes P \TrN(QR) &  \text{if \eeqref{IIId} holds}\\
        P\otimes S\TrN(QR)  &  \text{if \eeqref{IIIf} holds}                                     \end{cases}\Bigg\} \nonumber 
\\ &= \begin{cases}
      \TrN(QR) WVSTU\totimes P &  \text{if \eeqref{IVc} \& \eeqref{IIId} hold} \\
      \TrN(QR) WVPTU\totimes S &  \text{if \eeqref{IVc} \& \eeqref{IIIf} hold} \\
      \TrN(QR) VWSUT\totimes P &  \text{if \eeqref{IVd} \& \eeqref{IIId} hold} \\
      \TrN(QR) VWPUT\totimes S &  \text{if \eeqref{IVd} \& \eeqref{IIIf} hold} 
      \end{cases}
\end{align}
However, since the graph 
\begin{align*}
\includegraphicsd{.29}{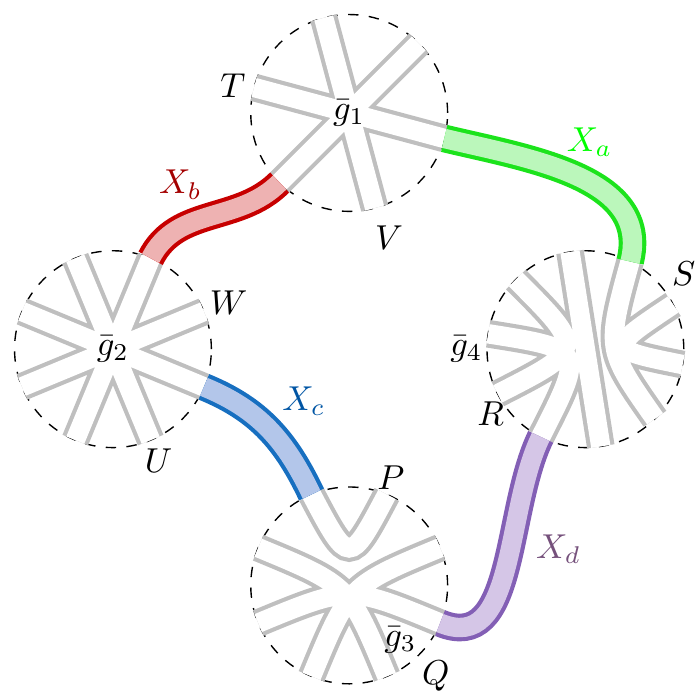} \subset  \Hess_{a,b} O_1 \star 
 \Hess_{b,c} O_2 \star \Hess_{c,d} O_3\star \Hess_{d,a} O_4
\end{align*}
has an $\A_n$-trace equal to $\TrN(QR) \times \TrN( WPUTSV)$,
only the last choice is possible. This proves at once 
\eeqref{IVd} and  \eeqref{IIIf}, that is 
\eeqref{FRGalgebraIV} and \eeqref{FRGalgebraIII}, respectively. 
\end{proof*}

\begin{remark} (On well-definedness of the graphical representation.) Example \ref{ex:cuts} 
shows a phenomenon that is more general: 
the asymmetry of the $M_n$-block structure of the nc Hessian matrix,  $\Hess_{a,b} \neq \Hess_{b,a}$. 
Nevertheless, a weaker symmetry persists. Since the swap of $X_a$ and $X_b$ in \eeqref{doubleNCder} 
leads to the exchange $\pi_1$ with $\pi_2$, we  conclude that for any interaction vertex $O$, \begin{align}\Hess_{a,b} O = \widetilde {\Hess_{b,a} O} \quad   \where \quad  \widetilde{(P\otimes Q)}= Q\otimes P\,\,\and\,\, \widetilde{(P\totimes Q)}= Q\totimes P\,, \label{HessNonsymm}\end{align}
for each $P, Q\in \Cn$; the exchange $P\totimes Q \to Q\totimes P$ follows
by Definition \ref{def:HessondoubleTrs}. This makes the present 
construction independent of the choice of ``inner'' and ``outer'' loop,
as well as the orientation of the interaction 
vertices in the one-loop (whether the Hessians 
of $O_1,O_2,\ldots,O_k$ being multiplied means 
we draw $\bar g_1,\ldots,\bar g_k$ clockwise or
anti-clockwise as in Fig. \ref{fig:notideal}) for the following reason. First, observe that using the algebra
obtained in Theorem \ref{thm:main} one can easily derive 
\begin{equation}
\widetilde{\mathfrak a\star \mathfrak b} = \tilde{\mathfrak b} \star \tilde{\mathfrak a } \qquad \text{ for each } \mathfrak a, \mathfrak b \in \A_n\,, \label{inversion} 
\end{equation} 
Let $\mathfrak h :=\Hess_{a_1,a_2} O_1 \star \Hess_{a_2,a_3} O_2\star \cdots \star
\Hess_{a_k,a_1} O_k $, for fixed $a_i=1,\ldots, n$ and for some fixed 
interaction vertices $O_i$ ($i=1,\ldots,k$). 
Then by \eeqref{HessNonsymm} 
\[ {\Hess_{a_1,a_k} O_k} \star \cdots   \star { \Hess_{a_3,a_2} O_2} \star 
  {\Hess_{a_2,a_1} O_1 }=
\widetilde{\Hess_{a_k,a_1} O_k}\star 
\widetilde{ \Hess_{a_2,a_3} O_2}\star \cdots
   \star \widetilde{\Hess_{a_1,a_2} O_1 } \,,\]
   and because of \eeqref{inversion}, this last 
   expression equals $\tilde{\mathfrak h}$. But according to 
   Lemma \ref{FRGtraces}, 
   $\TrA \tilde{\mathfrak{h}}=\TrA \mathfrak{h}$. 
   Thus, if $\mathfrak h $ contributes to the flow, so does $ \tilde{\mathfrak h} $,
and in an equal way, yielding our description independent on the cyclic orientation
we choose for drawing the interaction vertices. What we just proved can also be pictorially
justified:
\begin{equation}
\mathfrak h =
 \includegraphicsd{.45}{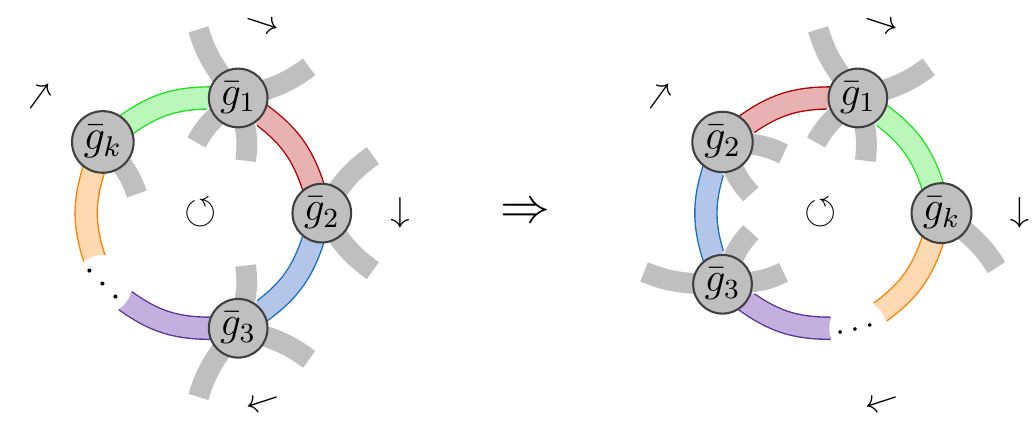} = \tilde{\mathfrak h}
\end{equation}
Notice that these are the only two representations that the cyclic orientation of each vertex $O_i$ allows (meaning, if 
one inverts the order of the interaction vertices).\par 
The final piece of well-definedness is that the product
found here is indeed associative, without using graphs. The purely algebraic proof is routine (and can be found in \cite{FRGEmultimatrix}).

\end{remark}

\begin{example}
Once proven the main statement, we can use the algebra to exemplify a
typical contribution to the renormalization flow in a Hermitian 3-matrix model.
Consider two operators $O_1=\frac{\bar g_1}{2} [\TrN(\frac{A^2}{2})]^2 $ and $O_2
= \bar g_2 \TrN(ABC)$. Suppose we wish to determine the $\bar g_1 \bar g_2$-coefficient of the
rhs of Wetterich equation. We need (essentially) the Hessian of $O_1$ and $[\Hess O_2]^{\star 2}$.
The former has only one non-zero entry,
\begin{align}\label{HessO1Ex}
\Hess_{I,J} O_1 & = \delta_I^J \delta_I^A \bar g_1\{
\underbrace{\TrN(A^2/2) [1_N\otimes 1_N]}_{ \includegraphicsds{.031}{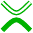} }+ \underbrace{A\totimes A}_{\includegraphicsds{.031}{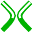}} \}
\,,
\end{align}
where a ``filled ribbon'' means that that half-edge is contracted in the one-loop graph,
and an ``empty ribbon'' that it is not (and therefore contributes to the final
effective vertex). We also have
\begin{align}
\Hess O_2=  \bar g_2
 \begin{bmatrix}
0 & C\otimes  1_N  & B\otimes 1_N \\
1_N\otimes C& 0 & A\otimes 1_N \\
1_N \otimes B &1_N \otimes A & 0
 \end{bmatrix}
 \end{align} getting   \begin{align}\label{HessO2Ex}
 [\Hess O_2]^{\star 2}=\bar g_2^2\begin{bmatrix}
C\otimes C + B\otimes B & \color{gray} B\otimes A & \color{gray}C\otimes A  \\
\color{gray} A\otimes B & \color{gray}A\otimes A + C\otimes C &\color{gray} C\otimes B \\
\color{gray}A\otimes C &\color{gray} B\otimes C &\color{gray} B\otimes B + A\otimes A
\end{bmatrix}
             \,.
\end{align}
Only the black-colored entry will contribute, since $\Hess O_1$'s (11)-entry is the only non-vanishing.
In the (11)-entry, the term 
\[ C\otimes C \text{ corresponds to the
graph $ \includegraphicsds{.08}{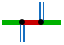}$ and $B\otimes B$
to  $ \includegraphicsds{.08}{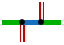}$} \] (the horizontal green
edges still to be contracted in the loop with those in Hessian \eqref{HessO1Ex} that are also filled).
Finally, we extract the coefficients \begin{align*}
[\bar g_1\bar g_2]  \STr \{ \Hess O_1 [\Hess O_2]^{\star 2}\} &= \TrA \big\{ [
\TrN(A^2/2) [1_N\otimes 1_N]+A\totimes A ] \star [C\otimes C + B\otimes B]
\big\} \\
&=  \TrA \big\{
\TrN(A^2/2) (C\otimes C + B\otimes B ) + A\totimes CAC + A\totimes BAB
\big\} \\ &= \TrN(A^2/2) \times [\TrN^2 C +\TrN^2 B   ] + \TrN (ACAC+ABAB )\,,
         \end{align*}
         which are effective vertices of the four  one-loop graphs that can be formed with the
         contractions of (the filled ribbon half-edges of)   \[ \text{ any of $ \Big\{\,\, \includegraphicsds{.07}{CtensorC} \,\,, \,\,\includegraphicsds{.07}{BtensorB}\,\, \Big\}$
         with any of $\Big\{ \,\, \includegraphicsds{.045}{AboxtimesA}\,\,, \,\,\includegraphicsds{.045}{TrA2_timesId}\,\,\Big\} $}\,.\]
 Of course, this is a toy-example: the algebraic structure pays off with
 higher-power interactions and/or higher number of matrices (whose flow becomes
 unaccessible by traditional methods and can hardly be cross-checked using graphs, due to the large amount of these; cf. the supplementary material of \cite{FRGEmultimatrix}).
 \end{example}

\section{Conclusion}
The algebraic structure of functional renormalization of Hermitian
$n$-matrix models with interactions containing several traces has been
addressed. Under the assumption that 
it is possible to compute the flow in 
 terms of $\mathrm U(N)$-invariant operators,
the present result completely describes the regulator-independent part of the flow.
This paper complements\footnote{To fully implement the flow for Dirac ensembles, which 
is the aim of \cite{FRGEmultimatrix} either operators
with broken unitary symmetry could be introduced. Another 
perspective is to implement a Ward-constrained flow
\cite{WardFlow_Review}} 
\cite{FRGEmultimatrix}. There, 
for multi-matrix models with multi-traces, Wetterich equation 
was proven, and in the middle of the proof one were able to read 
off the algebraic structure, 
\eqref{FRGalgebra}. Computations of $\beta$-functions using \eqref{FRGalgebra} revealed a one-loop structure in \cite{FRGEmultimatrix}.
Here we showed the converse: the one-loop structure requires the algebra of functional renormalization (i.e. the structure that makes the rhs of Wetterich equation computable for such matrix models) to be  \eeqref{FRGalgebra}, showing
its uniqueness. 
\par
As a final perspective, 
the present results can be useful to connect
different renormalization theories, e.g. \cite{Gurau:2010ts,TanasaVignes}. 
Also, Figure \ref{fig:ideal} is strikingly reminiscent of the Connes-Kreimer
residue defining the coproduct (of their renormalization Hopf-algebra \cite{ConnesKreimer}).
Between those, the algebraic language could build a shorter bridge
than graph theory---all the more, algebra 
can be coded more directly than graphs.

\small
\begin{acknowledgements} The two anonymous referees that received this article
in \textit{Letters in Mathematical Physics} are acknowledged for the
  extraordinarily careful reading and for useful comments.
    This work was mainly supported by the European Research Council (ERC) under
the European Union’s 
Horizon 2020 research and innovation program (grant agreement
No818066) and also by the
Deutsche Forschungsgemeinschaft (DFG, German Research Foundation)
under Germany’s 
Excellence Strategy EXC-2181/1-390900948 (the Heidelberg \textsc{Structures}
Cluster of Excellence). At the beginning, this project was 
funded by the TEAM programme of the
  Foundation for Polish Science co-financed by the European Union
  under the European Regional Development Fund
  (POIR.04.04.00-00-5C55/17-00).
\end{acknowledgements}

 \fontsize{10.4}{13.35}\selectfont    
\appendix
\section{On $\Gamma_N$ and $ R_N$}\label{app:RN}

Although our aim in this paper is rather algebraic,
we sketch how the FRG works, for sake of completeness.
Starting from the bare action  $S[\Phi]$ 
on $S:\H_N^n \to \re$, in order to construct the  effective action $\Gamma_N[\mathbb X]$
one first regulates the connected
partition function  $\log \mathcal Z [\mathbb{J}] = \log \int_{\H_N^n} \exp[-S[\Phi] +\sum_{c=1}^n \Tr( J_c \Phi_c)]  \dif   \Phi\Lebesgue$ 
(where $\mathbb{J}=(J_c)\in \H_N^n$)
by adding a mass-like term, $\frac12 r_{c,N}  \Phi^2_c$, to each matrix:
\begin{align*} \mathcal W_N [\mathbb{J}] = \log \int_{\H_N^n} \exp\bigg\{-S[\Phi] -\sum_{c=1}^n  \TrN\Big[\frac12 r_{c,N}  \Phi^2_c+ J_c \Phi_a\Big]\bigg\}  \dif \Phi\Lebesgue\end{align*} where $r_{c,N}$ is the next \textit{infrared regulator}, given in terms of the Heaviside or indicator function $\Theta_{\mathbb D_N}$ on the disk $\mathbb D_N=\{ (a,b)\in \N^2 \mid  a^2+b^2\leq N^2\}$ by  \vspace{-.5cm}
\begin{align}
r_{c,N} (a,b)=  Z_c\cdot    \bigg[\frac{N^2}{ a^2+b^2} -1 \bigg] \cdot \Theta_{\mathbb D_N}(a,b) \,, \hspace{.41cm}\text{or plotted: }
\includegraphicsd{.294}{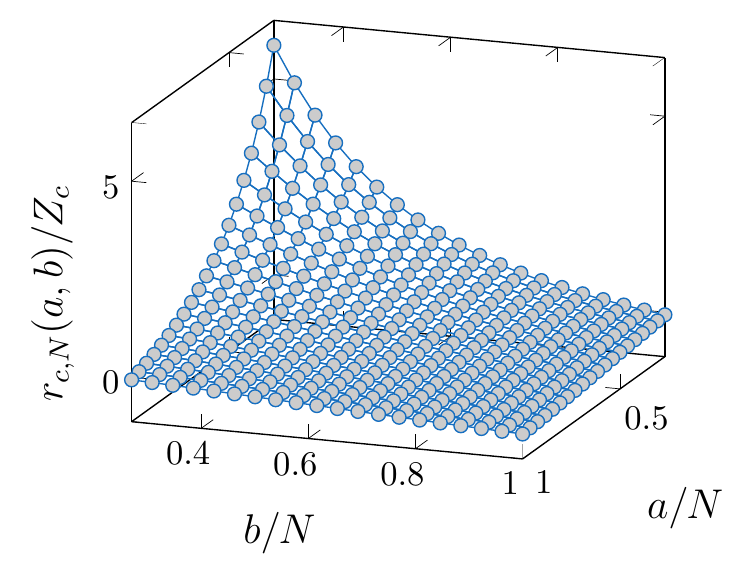}
\end{align}
with $Z_c$ the ``wave function renormalization'' of the matrix $X_c$  (see \cite{FRGEmultimatrix} for details on the dependence on the ``classical fields'' $\mathbb X=(X_1,\ldots,X_n):=(\partial_{J_1} \mathcal W,\ldots,\partial_{J_n} \mathcal W) \in \H_N^n $).  
Other  regulators are not discussed here, but another regulator should comply with the following 
limits: in order for $r_{c,N}$ to integrate out only the ``higher modes'', i.e. matrix entries $\big\{(X_c)_{i,j}\big\}_{i,j\geq N; c=1,\ldots,n}$ above the \textit{energy scale} $N$, one imposes that $r_{c,N}(a,b)=0 $ if $a>N$ or $b>N$, which explains the presence of $\Theta_{\mathbb D_N}$ in this particular choice. The additional condition $r_{c,N}>0$
creates a mass-like term for the low-energy modes $\big\{ (X_c)_{i,j}\big\}_{0<i,j< N; c=1,\ldots,n}$
that protects these from being integrated out. 
One wants, moreover, to recover the bare action via sadle point approximation as $N\to\infty$ and thus   
$r_{c,N} \to \infty$ is necessary in that limit. The interpolating  \textit{effective action}
is then constructed by taking the Legendre transform, namely $\Gamma_N[\mathbb X]= \sup_{J_1,\ldots,J_n} \sum_{c=1}^n\big\{ \TrN(X_c J_c) -\mathcal W_N[J_1,\ldots,J_n] -\frac12 \TrN(r_{c,N} X_c^2)\big\}$.
The regulator $R_N$ that appears in the main text is
\[R_N  = \begin{bmatrix}
      r_{1,N} 1_N\otimes 1_N  &  0&   \ldots & 0 \\
      0 &  r_{2,N} 1_N\otimes 1_N & \ldots & 0 \\ 
      \vdots & \ddots & & \vdots \\[1ex]
       0& 0& \ldots & r_{n,N} 1_N\otimes 1_N
      \end{bmatrix}.\] 
 
\providecommand{\href}[2]{#2}\begingroup\raggedright\endgroup 
 
\end{document}